\documentclass[11pt,reqno]{amsproc}

\usepackage{amsfonts,amssymb,amsmath,gensymb,color,epic,eepic,float,fullpage}
\usepackage{empheq}
\usepackage{enumitem}
\usepackage[mathscr]{euscript}
\usepackage[T1]{fontenc}
\usepackage{mathtools}
\usepackage{comment}
\usepackage{bbm,bm}
\usepackage{caption}
\usepackage{hyperref}
\usepackage{tikz}
\usetikzlibrary{graphs,patterns,decorations.markings,arrows,matrix}
\usetikzlibrary{calc,decorations.pathmorphing,decorations.markings,
decorations.pathreplacing,patterns,shapes,arrows}
\tikzstyle{block} = [rectangle,draw,text width=10em,text centered,rounded corners,minimum height=4em]
\tikzstyle{line} = [draw, -latex']
\allowdisplaybreaks

\theoremstyle{plain}
\newtheorem{defn}{\bf{\textsc{Definition}}}[section]

\newtheorem{thm}[defn]{\bf{\textsc{Theorem}}}
\newtheorem{conj}[defn]{\bf{\textsc{Conjecture}}}
\newtheorem{cor}[defn]{\bf{\textsc{Corollary}}}
\newtheorem{remark}[defn]{\bf{\textsc{Remark}}}
\newtheorem{prop}[defn]{\bf{\textsc{Proposition}}}

\def\leq{\leqslant}
\def\geq{\geqslant}

\def\i{{\sf i}}

\def\S{\mathcal{S}}

\def\phid{\phi^\dagger}

\def\ie{{\it i.e.}\/}
\def\eg{{\it e.g.}\/}

\def\be{\begin{equation}}
\def\ee{\end{equation}}

\newcommand{\ket}[1]{\left|#1\right\rangle}

\begin{document}

\title[Stochastic duality and deformed KZ equations]{Integrable stochastic dualities and the deformed Knizhnik--Zamolodchikov equation}

\author{Zeying Chen, Jan~de~Gier and Michael~Wheeler}
\address{Australian Research Council Centre of Excellence for Mathematical and Statistical Frontiers (ACEMS), School of Mathematics and Statistics, The University of Melbourne, VIC 3010, Australia}
\email{zeyingc@student.unimelb.edu.au, jdgier@unimelb.edu.au, wheelerm@unimelb.edu.au}

\maketitle

\begin{abstract}
We present a new method for obtaining duality functions in multi-species asymmetric exclusion processes (mASEP), from solutions of the deformed Knizhnik--Zamolodchikov equations. Our method reproduces, as a special case, duality functions for the self-dual single species ASEP on the integer lattice.
\end{abstract}

\section{Introduction}
\label{se:intro}

\subsection{Background}

Duality plays an important role in stochastic Markov processes where the time evolution is described by a linear generator. Early applications appear in \cite{spitzer} for the self-dual symmetric exclusion process, and in \cite{harris} for the contact process. Apart from these classical applications, duality is also a valuable tool for proving the limits of particle systems to stochastic partial differential equations; see \cite{CorwinT,CorwinST}. 

A duality functional of two processes is an observable that co-varies in time with respect to the evolution of the processes; see for example \cite{JansenK,Liggett}. Duality functionals are most powerful when expectation values and correlation functions of many-particle processes are related to those containing few particles. Models with few particles can be analysed in great detail and therefore expectation values can often be calculated analytically via such dualities. A well-known recent example is that of the duality between the stochastic Kardar--Parisi--Zhang (KPZ) equation for interface growth \cite{kpz} and the integrable one-dimensional quantum Bose gas \cite{brunetDerrida,CalabreseLD,kardar}. Indeed, much progress has been made in recent years using duality in the setting of integrable stochastic processes such as \cite{BorodinC,bcs,CorwinP,GiardinaKR,ImamuraS11}, where several powerful tools are available.

In many cases treated in the literature, duality functionals have been constructed in a more or less {\it ad hoc} fashion and only a few attempts have been made to systematically derive dualities in integrable stochastic models using quantum group symmetries \cite{belschuetz,carinciGRSa,carinciGRSb,kuan,kuan2,schuetzSandow}. In this paper we propose a new approach for methodically constructing integrable dualities by exploiting the algebraic structure provided by the $t$-deformed Knizhnik--Zamolodchikov (KZ) equations \cite{frenkel-resh,KZ}, which are consistency equations expressed in terms of the R-matrix of a quantum group, or alternatively, in terms of the Hecke algebra.

We will work in the context of the integrable (multi-species) asymmetric exclusion simple process (mASEP) with hopping rate $t$. The mASEP can be realized in two ways via representations of the Hecke algebra. The first is a standard description in which each particle configuration $\mu$ is identified with a basis element of a vector space, and where the local Markov generator is a matrix acting on this space. The second realization is on a basis $f_\mu$ of the ring of $n$-variable polynomials, in which the local Markov generator becomes a divided-difference operator (a polynomial representation of a Hecke generator). The $t$-deformed KZ equations connect these two realizations, and can in turn be interpreted as the duality relations of a diagonal observable intertwining the vector space and polynomial representations of the mASEP. 

In order to go beyond this tautological diagonal observable, and obtain non-trivial observables on the two processes, our main technical tool will be a family of $n$-variable polynomials $f_\mu$ studied in \cite{CantiniGW}. These polynomials are a standard basis for the polynomial realization of the mASEP, and are closely related to the theory of symmetric Macdonald polynomials \cite{Macda,Macdb} and their non-symmetric versions \cite{chereda,cheredb,opdam}. The $f_{\mu}$ polynomials depend on two parameters: the mASEP hopping parameter $t$, and another parameter $q$ which appears when imposing a certain cyclic boundary condition; collectively, these parameters are the $(q,t)$ of Macdonald polynomial theory. The presence of the second parameter $q$ is crucial to our approach, for while it has no direct physical meaning in the mASEP, its value can be tuned. In particular, when the $(q,t)$ parameters satisfy a resonance condition of the form,
\begin{align}
\label{resonance}
q^{k} t^{l}=1,\qquad  k,l \in\mathbb{N},
\end{align}
the $f_{\mu}$ polynomials may become singular and (after appropriately normalizing, to remove poles) degenerate into a sum of the form $\sum_{\nu} \psi(\nu,\mu;t) f_{\nu}$, for certain coefficients $\psi(\nu,\mu;t)$. In other words, the condition \eqref{resonance} creates linear dependences between the $f_{\mu}$ polynomials and thus gives rise to non-trivial intertwining solutions of the $t$-deformed KZ equations. It is these solutions that produce duality relations in the mASEP; the duality functionals end up being nothing but (rescaled versions of) the expansion coefficients $\psi(\nu,\mu;t)$.

In the rest of the introduction, we describe our methodology in greater detail.

\subsection{Functional definition of duality}
\label{sec:functional}

The standard definition of a stochastic duality is in terms of a function $\psi$ which takes values on the configuration spaces of two (possibly different) Markov processes. Let us begin by restating this definition in some generality.

Let $\mathbb{A}$ and $\mathbb{B}$ be two (possibly infinite) sets, whose elements we denote by $a$ and $b$, respectively. Let $\mathbb{F}$ be the space of all functions $\psi$ of the form
\begin{align*}
&
\psi: \mathbb{A} \times \mathbb{B} \rightarrow \mathbb{C}.
\end{align*}
Consider two linear functionals $L$ and $M$ which act on functions in $\mathbb{F}$ as follows:
\begin{align}
\label{functionals}
L\left[\psi(\cdot,b)\right] (a) := \sum_{a' \in \mathbb{A}} \ell(a,a') \psi(a',b),
\quad\quad
M \left[\psi(a,\cdot)\right] (b) := \sum_{b' \in \mathbb{B}} m(b,b') \psi(a,b'),
\end{align}
where $\ell:\mathbb{A} \times \mathbb{A} \rightarrow \mathbb{C}$ and $m:\mathbb{B} \times \mathbb{B} \rightarrow \mathbb{C}$ are some pre-specified functions (in the language of stochastic processes, these will be the matrix entries of the Markov generators $L$ and $M$ of two different processes). Then $L$ and $M$ are {\it dual} with respect to a function $\psi$ if
\begin{align}
\label{func-dual}
L \left[\psi(\cdot,b)\right] (a)
=
M \left[\psi(a,\cdot)\right] (b),
\quad\quad
\forall\ a \in \mathbb{A},\ b \in \mathbb{B}.
\end{align}

\subsection{Matrix definition of duality}
\label{sec:matrix-dual}

It is useful for our purposes to recast the statement of duality in terms of matrices, rather than functionals. We upgrade the previous sets $\mathbb{A}$ and $\mathbb{B}$ to vector spaces, with basis vectors $\ket{a}$ and $\ket{b}$. Let $\psi$ be a certain function in $\mathbb{F}$ and consider the following vector, $\ket{\Psi} \in\mathbb{A} \otimes \mathbb{B}$:
\begin{align}
\label{Psidef}
\ket{\Psi}
:=
\sum_{\substack{a \in \mathbb{A} \\ b \in \mathbb{B}}}
\psi(a,b)\ket{a} \otimes \ket{b}.
\end{align}
Let $\mathbb{L} \in {\rm End} (\mathbb{A})$ and $\mathbb{M} \in {\rm End}(\mathbb{B})$ be linear operators given explicitly by
\begin{align}
\label{lin-ops}
\mathbb{L} \ket{a} = \sum_{a' \in \mathbb{A}} \ell(a',a) \ket{a'},
\quad\quad
\mathbb{M} \ket{b} = \sum_{b' \in \mathbb{B}} m(b',b) \ket{b'},
\end{align}
for certain matrix entries $\ell$ and $m$.

\begin{prop}
\label{lem1}
The duality relation \eqref{func-dual} is equivalent to the equation
\begin{align}
\label{mat-dual}
\mathbb{L} \ket{\Psi}
=
\mathbb{M}\ket{\Psi}.
\end{align}
\end{prop}

\begin{proof}
Explicit calculation of the left and right hand sides gives
\begin{align*}
\mathbb{L} \ket{\Psi}
&=
\sum_{a,b,a'}
\psi(a,b) \ell(a',a)
\ket{a'} \otimes \ket{b}
=
\sum_{a,b} \left( \sum_{a'} \ell(a,a')  \psi(a',b) \right) \ket{a} \otimes \ket{b},
\\
\mathbb{M} \ket{\Psi}
&=
\sum_{a,b,b'}
\psi(a,b) m(b',b)
\ket{a} \otimes \ket{b'}
=
\sum_{a,b} \left( \sum_{b'} m(b,b')  \psi(a,b') \right) \ket{a} \otimes \ket{b}.
\end{align*}
Requiring that these be equal implies \eqref{func-dual} for the function $\psi$.
\end{proof}

\subsection{$t$KZ equations as a source of dualities}
\label{se:tKZ-source}

The local $t$-deformed Knizhnik--Zamolodchikov equations,\footnote{We use the term {\it local} to distinguish these equations from the original quantum deformation of the Knizhnik--Zamolodchikov equation introduced by Frenkel and Reshetikhin \cite{frenkel-resh}, which involves global scattering matrices. Our use of $t$ rather than $q$ as the deformation parameter stems from the fact that both parameters play a role in this work, as the $(q,t)$ in Macdonald polynomials.} or $t$KZ equations for short, as introduced by Smirnov in the study of form factors \cite{smirnov}, are a system of equations for a polynomial-valued\footnote{In many contexts solutions to the $t$KZ equations are in fact in terms of series and elliptic functions.} vector $\ket{\Psi} \in \mathbb{C}[z_1,\dots,z_n] \otimes \mathbb{V}$. Here $\mathbb{C}[z_1,\dots,z_n]$ denotes the ring of polynomials in $n$ variables $(z_1,\dots,z_n)$, over the field of complex numbers. The vector space $\mathbb{V}$ is obtained by taking an $n$-fold tensor product of local spaces, \ie\ $\mathbb{V} := \mathbb{V}_1 \otimes \cdots \otimes \mathbb{V}_n$, where $\mathbb{V}_i \equiv \mathbb{C}^{r+1}$ for all $1 \leq i \leq n$, and $r \geq 1$ is some fixed positive integer. The local $t$KZ equations read
\begin{align}
\label{general-qkz}
s_i \ket{\Psi} = \check{R}(z_i / z_{i+1}) \ket{\Psi},
\quad
i \in \{1,\dots,n-1\},
\end{align}
where $s_i$ is a simple transposition acting on $\mathbb{C}[z_1,\dots,z_n]$, with action
\begin{align*}
s_i g(z_1,\dots,z_i,z_{i+1},\dots,z_n) = g(z_1,\dots,z_{i+1},z_i,\dots,z_n),
\quad
\forall\ g \in \mathbb{C}[z_1,\dots,z_n],
\end{align*}
and $\check{R}(z_i / z_{i+1})$ denotes the R-matrix associated to quantized affine $\mathfrak{sl}(r+1)$ acting in $\mathbb{V}_i \otimes \mathbb{V}_{i+1}$. To fix a particular solution of \eqref{general-qkz} these equations are supplemented by a cyclic boundary condition on $\ket{\Psi}$, which we do not write down at this stage.

It is known (see for example \cite{Pasquier,Zinn-Justin}) that the equations \eqref{general-qkz} can be cast in the form
\begin{align}
\label{local-qkz}
\mathbb{L}_i \ket{\Psi}
=
\mathbb{M}_i \ket{\Psi},
\quad
i \in \{1,\dots,n-1\},
\end{align}
for certain $\mathbb{L}_i \in {\rm End}(\mathbb{C}[z_1,\dots,z_n]) \otimes 1$ and $\mathbb{M}_i \in 1 \otimes {\rm End}(\mathbb{V})$. This form differs slightly from \eqref{general-qkz}, since it separates completely the action on the $\mathbb{C}[z_1,\dots,z_n]$ part of $\ket{\Psi}$ from that on its $\mathbb{V}$ part. The equation \eqref{local-qkz} is our key to establishing the link between $t$KZ equations and dualities. The connection can be made precise under the following steps:
\begin{itemize}
\item We identify the two generic vector spaces appearing in Section \ref{sec:matrix-dual} with the vector spaces appearing in \eqref{local-qkz}, \ie\ $\mathbb{A} \equiv \mathbb{C}[z_1,\dots,z_n]$ and $\mathbb{B} \equiv \mathbb{V}$.

\item We choose suitable bases $\{\ket{a}\}$ and $\{\ket{b}\}$ for $\mathbb{A}$ and $\mathbb{B}$, and expand both $\ket{\Psi}$ and the linear operators $\mathbb{L}_i$ and $\mathbb{M}_i$ with respect to these bases, as in \eqref{Psidef} and \eqref{lin-ops}. This yields
\begin{align*}
\sum_{a,a' \in \mathbb{A}}
\sum_{b \in \mathbb{B}}
\ell_i(a,a')
\psi(a',b)
\ket{a} \otimes \ket{b}
=
\sum_{a \in \mathbb{A}}
\sum_{b,b' \in \mathbb{B}}
m_i(b,b')
\psi(a,b')
\ket{a} \otimes \ket{b},
\end{align*}
in the very same way as in the proof of Proposition \ref{lem1}.

\item The coefficients $\psi(a,b)$ are then duality functions\footnote{In the rest of the paper we will refer to such coefficients as duality {\it functions} rather than {\it functionals}. The reason for this is that we only focus on $\psi$ as a function on the underlying configuration spaces, and suppress the fact that configurations $a$ and $b$ are themselves functions of time.} with respect to $n-1$ pairs of linear functionals $L_i$ and $M_i$, in the same sense as \eqref{func-dual}:
\begin{align}
\label{sum_i}
\sum_{a' \in \mathbb{A}}
\ell_i(a,a')
\psi(a',b)
=
\sum_{b' \in \mathbb{B}}
m_i(b,b')
\psi(a,b'),
\end{align}
where $\ell_i(a,a')$ and $m_i(b,b')$ are the matrix entries of the operators $\mathbb{L}_i$ and $\mathbb{M}_i$. The $\psi(a,b)$ can also be thought of as duality functions with respect to the generators $L := \sum_{i=1}^{n-1} L_i$ and $M := \sum_{i=1}^{n-1} M_i$, simply by summing \eqref{sum_i} over $1 \leq i \leq n-1$.
\end{itemize}
This procedure allows one, in principle, to start from any polynomial solution of the local relations \eqref{general-qkz} and to extract from it duality functions. However, it cannot be applied without due heed to the particulars of the solution that one chooses. For example, finding bases $\{\ket{a}\}$ and $\{\ket{b}\}$ such that the operators $\mathbb{L}_i$ and $\mathbb{M}_i$ are meaningful as Markov matrices may be quite difficult in practice or not even possible. It is also not guaranteed that the functions $\psi(a,b)$ define an interesting statistic on the two configuration spaces $\mathbb{A}$ and $\mathbb{B}$. In this paper, we will recover a known interesting statistic from a specific solution of \eqref{general-qkz} which was previously considered in \cite{CantiniGW,GierW}.

\subsection{Notation and conventions}

Let us outline some of the notation to be used in the paper. A composition $\mu$ is an $n$-tuple of non-negative integers, $(\mu_1,\dots,\mu_n)$. The elements of $\mu$, $\mu_i \geq 0$, are referred to as parts. We define the part-multiplicity function $m_i(\mu)$ as the number of parts in $\mu$ equal to $i$: $m_i(\mu) = \{k: \mu_k=i\}$. A partition $\lambda$ is a composition with weakly decreasing parts, $(\lambda_1 \geq \cdots \geq \lambda_n \geq 0)$. We also define anti-partitions $\delta$, which are compositions with weakly increasing parts, $(0 \leq \delta_1 \leq \cdots \leq \delta_n)$. Where possible we reserve the letters $\mu,\nu$ for generic compositions, $\lambda$ for partitions, and $\delta$ for anti-partitions. Given a composition $\mu$, its (anti-)dominant ordering ($\mu^{-}$) $\mu^{+}$ is the unique (anti-)partition obtainable by permuting the parts of $\mu$.

At times we will consider compositions of infinite length. By this, we shall always mean finitely-supported infinite strings $(\dots,\mu_{-1},\mu_0,\mu_{1},\dots)$, where $\mu_i \geq 0$ for all $i \in \mathbb{Z}$ and where there exists $N$ such that $\mu_i = 0$ if $|i| > N$.

Following \cite{KasataniT}, we define two orders on compositions. The first is the dominance order, denoted by $\geq$. Given two compositions $\mu = (\mu_1,\dots,\mu_n)$ and $\nu = (\nu_1,\dots,\nu_n)$, we define
\begin{align*}
\mu \geq \nu \iff \sum_{i=1}^{j} \mu_i \geq \sum_{i=1}^{j} \nu_i, \qquad \forall\ 1 \leq j \leq n.
\end{align*}
The second order is denoted by $\succ$. Given two compositions $\mu$ and $\nu$, we define
\begin{align*}
\mu \succ \nu \iff
\Big(
\mu^{+} > \nu^{+}
\quad
\text{or}
\quad
\mu^{+} = \nu^{+},
\
\mu > \nu
\Big).
\end{align*}
This order should not be confused with the interlacing of partitions, which is another standard use of the symbol $\succ$ in the literature.

We let $\mathbb{C}_{q,t}[z_1,\dots,z_n]$ denote the ring of polynomials in $(z_1,\dots,z_n)$ with coefficients in $\mathbb{Q}(q,t)$. We use the shorthand $z^{\mu} := z_1^{\mu_1} \dots z_n^{\mu_n}$ to denote the elements of the monomial basis. Given a polynomial $g(z_1,\dots,z_n) \in \mathbb{C}_{q,t}[z_1,\dots,z_n]$, $p \in \mathbb{N}$ and $m \in \mathbb{Q}_{>0}$, we define
\begin{align*}
{\rm Coeff}_p[g,m] := \lim_{q \rightarrow t^{-m}} (1-qt^{m})^p g(z_1,\dots,z_n),
\end{align*}
where the limit exists. In this work we are mainly interested in simple poles in $q$, when it is convenient to write ${\rm Coeff}_1[g,m] \equiv {\rm Coeff}[g,m]$. For two polynomials $g_1, g_2 \in \mathbb{C}_{q,t}[z_1,\dots,z_n]$, we write
\begin{align*}
g_1(z_1,\dots,z_n) \propto g_2(z_1,\dots,z_n) \iff \exists\ \alpha \in \mathbb{Q}(q,t)\ \text{such that}\
g_1(z_1,\dots,z_n) = \alpha g_2(z_1,\dots,z_n).
\end{align*}

\subsection{Acknowledgments}

We gratefully acknowledge support from the Australian Research Council Centre of Excellence for Mathematical and Statistical Frontiers (ACEMS). MW is supported by an Australian Research Council DECRA. It is a pleasure to thank Alexei Borodin, Ivan Corwin, Alexandr Garbali, Jeffrey Kuan, Tomohiro Sasamoto and Ole Warnaar for their interest in this work and for related discussions.

\section{Asymmetric simple exclusion process}
\label{se:asep}

The functional definition of duality \eqref{func-dual}, and its matrix version \eqref{mat-dual}, are both generic statements that apply for any indexing sets $\mathbb{A}$ and $\mathbb{B}$. In this section we will show how self-duality in the ASEP can be cast within this general framework, forming the foundations of the rest of the paper.

In the examples of duality in ASEP in \cite{bcs}, duality is exhibited between two different ASEP systems (which contain different numbers of particles, and different hopping rates) on the infinite line. This means that we should expect both $\mathbb{A}$ and $\mathbb{B}$ to be identified with the set of infinite binary strings. More concretely, we shall define $\mathbb{A}$ to be the space of all multilinear polynomials in an infinite set of variables $\{z\} = \{\dots,z_{-1},z_0,z_1,\dots\}$. The basis vectors of this space are $\prod_{i \in \mathbb{Z}} z_i^{\nu_i}$, where $\nu$ is an infinite composition with $\nu_i \in \{0,1\}$ for all $i \in \mathbb{Z}$. The binary string corresponding with a given basis vector is read off simply as the exponents of the variables $\{z\}$. On the other hand, we define $\mathbb{B}$ to be the infinite tensor product $\bigotimes_{i \in \mathbb{Z}} \mathbb{C}_i^2$ whose basis vectors are $\bigotimes_{i \in \mathbb{Z}} \ket{\mu_i}_i$, where $\mu$ is an infinite composition with $\mu_i \in \{0,1\}$ for all $i \in \mathbb{Z}$, and where $\ket{0}$ and $\ket{1}$ denote the canonical basis of $\mathbb{C}^2$.

\subsection{The ASEP generators $L_i$ and $M_i$}
\label{sec:generator}

Here we recall the definition of the ASEP generator, denoting it $L$, to match the notation of Section \ref{sec:functional}. It is constructed as a sum of local generators, $L = \sum_{i \in \mathbb{Z}} L_i$. Each local generator $L_i$ acts on functions $\psi$ of binary strings $\nu$. Particles (the ones of the binary string) hop to the left with rate $1$ and to the right with rate $t$:
\begin{align}
\label{Mi}
L_i [\psi](\nu) = \sum_{\nu' \in \mathbb{A}} \ell_i(\nu,\nu') \psi(\nu'),
\end{align}
where the coefficients $\ell_i(\nu,\nu')$, which specify the transition rate from $\nu$ to $\nu'$, are given by
\begin{align}
\label{pos}
\ell_i(\nu,\nu')
=
\left\{
\begin{array}{rl}
t,
&
\nu_i > \nu_{i+1},
\quad
(\nu_i,\nu_{i+1}) = (\nu'_{i+1},\nu'_i),
\quad
\nu_k = \nu'_k\ \forall\ k \not= i,i+1,
\\ \\
1,
&
\nu_i < \nu_{i+1},
\quad
(\nu_i,\nu_{i+1}) = (\nu'_{i+1},\nu'_i),
\quad
\nu_k = \nu'_k\ \forall\ k \not= i,i+1,
\\ \\
0,
&
\text{otherwise},
\end{array}
\right.
\end{align}
when $\nu \not= \nu'$, and where the diagonal elements are chosen such that the matrix {\it rows} sum to zero:
\begin{align}
\label{diag}
\ell_i(\nu,\nu)
=
\left\{
\begin{array}{rl}
-t,
&
\nu_i > \nu_{i+1},
\\ \\
-1,
&
\nu_i < \nu_{i+1},
\\ \\
0,
&
\text{otherwise}.
\end{array}
\right.
\end{align}
Similarly, one can define a reverse ASEP generator whose hopping rates have been switched, \ie\ particles now hop to the left with rate $t$ and to the right with rate $1$. We shall denote this generator by $M = \sum_{i \in \mathbb{Z}} M_i$, again in reference to our notation in Section \ref{sec:functional}. It acts on functions $\psi$ of binary strings $\mu$:
\begin{align}
\label{Li}
M_i [\psi](\mu) = \sum_{\mu' \in \mathbb{B}} m_i(\mu,\mu') \psi(\mu'),
\end{align}
where the hopping rates are given by
\begin{align}
\label{pos-back}
m_i(\mu,\mu')
=
\left\{
\begin{array}{rl}
1,
&
\mu_i > \mu_{i+1},
\quad
(\mu_i,\mu_{i+1}) = (\mu'_{i+1},\mu'_i),
\quad
\mu_k = \mu'_k\ \forall\ k \not= i,i+1,
\\ \\
t,
&
\mu_i < \mu_{i+1},
\quad
(\mu_i,\mu_{i+1}) = (\mu'_{i+1},\mu'_i),
\quad
\mu_k = \mu'_k\ \forall\ k \not= i,i+1,
\\ \\
0,
&
\text{otherwise},
\end{array}
\right.
\end{align}
when $\mu \not= \mu'$, and where the diagonal elements are chosen such that the matrix {\it columns} sum to zero:
\begin{align}
\label{diag-back}
m_i(\mu,\mu)
=
\left\{
\begin{array}{rl}
-t,
&
\mu_i > \mu_{i+1},
\\ \\
-1,
&
\mu_i < \mu_{i+1},
\\ \\
0,
&
\text{otherwise}.
\end{array}
\right.
\end{align}
The linear operators $\mathbb{L}_i$ and $\mathbb{M}_i$ with matrix entries $\ell_i(\nu,\nu')$ and $m_i(\mu,\mu')$ can be turned into Markov matrices by addition of the identity matrix. Following the standard conventions of the probability literature, $\mathbb{L}_i$ acts to the left, while $\mathbb{M}_i$ acts to the right. However, since we intend to cast $\mathbb{L}_i$ as an operator on the space of polynomials (as explained in the next section), we find that left-action is notationally cumbersome, and instead arrange so that both $\mathbb{L}_i$ and $\mathbb{M}_i$ act to the right.

\subsection{Divided-difference realization of $\mathbb{L}_i$}
\label{sec:asep-l_i}

Let $\mathbb{A}$ denote the space of multilinear polynomials in $\{\dots,z_{-1},z_0,z_1,\dots\}$, and let us seek an operator
$\mathbb{L}_i$ whose action on $\mathbb{A}$ faithfully reproduces \eqref{lin-ops} with coefficients given by \eqref{pos}--\eqref{diag}. We define a linear operator $\mathbb{L}_i$ on $\mathbb{A}$ by
\begin{align}
\label{asepL}
\mathbb{L}_i
=
\left(
\frac{t z_i - z_{i+1}}{z_i - z_{i+1}}
\right)
(s_i -1),
\end{align}
where we recall that $s_i$ acts on polynomials by the simple transposition $z_i \leftrightarrow z_{i+1}$.

\begin{prop}
Let $\nu$ be a binary string and associate to it the monomial $\ket{\nu} = \prod_{i \in \mathbb{Z}} z_i^{\nu_i}$. Then $\mathbb{L}_i \ket{\nu} = \sum_{\nu' \in \mathbb{A}} \ell_i(\nu',\nu) \ket{\nu'}$, where the expansion coefficients are given by \eqref{pos}--\eqref{diag}.
\end{prop}

\begin{proof}
It is easy to check that $\mathbb{L}_i$ has a stable action on the space of multilinear polynomials in $\{z\}$, meaning that we can indeed expand $\mathbb{L}_i \ket{\nu}$ on this space. Furthermore it is clear from its definition that $\mathbb{L}_i$ only acts non-trivially on the variables $(z_i,z_{i+1})$, meaning that there are only three cases to check:
\begin{align}
\label{action-L_i}
\mathbb{L}_i \left( \prod_{k \in \mathbb{Z}} z_k^{\nu_k} \right)
=
\prod_{\substack{k \in \mathbb{Z} \\ k \not=i,i+1}} z_k^{\nu_k}
\times
\left\{
\begin{array}{ll}
0, & \nu_i = \nu_{i+1},
\\
(z_{i+1} - tz_i), & \nu_i > \nu_{i+1},
\\
(tz_i - z_{i+1}), & \nu_i < \nu_{i+1},
\end{array}
\right.
\end{align}
where the vanishing of the first case is due to the fact that $\mathbb{L}_i$ annihilates any polynomial which is symmetric in $(z_i,z_{i+1})$. The coefficients obtained from \eqref{action-L_i} directly match those in \eqref{pos}--\eqref{diag}.

\end{proof}

\subsection{Matrix realization of $\mathbb{M}_i$}
\label{sec:asep-m_i}

Let $\mathbb{B} = \bigotimes_{i \in \mathbb{Z}} \mathbb{C}_i^2$ and construct basis vectors $\ket{\mu} = \bigotimes_{i \in \mathbb{Z}} \ket{\mu_i}_i$, where each $\mu_i$ takes values in $\{0,1\}$ and
\begin{align*}
\ket{0} = \binom{1}{0},
\quad\quad
\ket{1} = \binom{0}{1}.
\end{align*}
Let $\mathbb{M}_i$ be the linear operator on $\mathbb{B}$ which acts according to \eqref{lin-ops}, with matrix elements given by \eqref{pos-back}--\eqref{diag-back}. We see that
\begin{align}
\label{asepM}
\mathbb{M}_i
=
\left(
\begin{array}{cccc}
0 & 0 & 0 & 0
\\
0 & -1 & +t & 0
\\
0 & +1 & -t & 0
\\
0 & 0 & 0 & 0
\end{array}
\right)_{i,i+1}
\end{align}
where the subscript indicates that the matrix acts non-trivially only on the spaces $\mathbb{C}^2_i$ and $\mathbb{C}^2_{i+1}$ of the tensor product, acting as the identity on all other spaces.

\subsection{Local duality relation}

Now we come to the formulation of duality in the ASEP. We say that $\psi$ is a {\it local ASEP duality function} provided that, for all $i \in \mathbb{Z}$,
\begin{align}
\label{local-mat-dual}
\mathbb{L}_i \ket{\Psi} = \mathbb{M}_i \ket{\Psi},
\quad\quad
\text{where}\ \
\ket{\Psi}
=
\sum_{\nu \in \mathbb{A}}
\sum_{\mu \in \mathbb{B}}
\psi(\nu,\mu) \prod_{k \in \mathbb{Z}} z_k^{\nu_k} \ket{\mu}.
\end{align}
As we already showed in Section \ref{sec:matrix-dual}, this then implies that $\psi$ satisfies the functional version of duality
\begin{align*}
L_i[\psi(\cdot,\mu)](\nu) = M_i[\psi(\nu,\cdot)](\mu),
\quad
\forall \ i \in \mathbb{Z},
\end{align*}
with respect to the local ASEP generators \eqref{Mi} and \eqref{Li}. It is clear that any local duality function $\psi$ will also be a duality function with respect to the global generators $L = \sum_{i \in \mathbb{Z}} L_i$ and $M = \sum_{i \in \mathbb{Z}} M_i$, however the converse is not necessarily true. In the rest of the paper we will focus on obtaining non-trivial solutions of \eqref{local-mat-dual} and its higher-rank analogue \eqref{local-mat-dual2}, even though we cannot {\it a priori} expect to obtain all possible global duality functions in this way.

\subsection{Generalization to multi-species ASEP}

All of the notions considered so far admit an extension to the multi-species ASEP. The mASEP is a continuous-time Markov chain of hopping coloured particles, \ie\ it is defined on general strings of non-negative integers, or compositions. In order to study it in our framework, we now identify $\mathbb{A}$ and $\mathbb{B}$ with the set of infinite compositions. We will assume that the parts of these compositions are bounded by some $r \in \mathbb{N}$, where $r$ denotes the number of particle species present in the mASEP under consideration. The ordinary ASEP is recovered by choosing $r = 1$.

The local mASEP generators $L_i$ and $M_i$ are given by the very same formulae as in Section \ref{sec:generator}, \ie\ by the equations \eqref{Mi}--\eqref{diag} and \eqref{Li}--\eqref{diag-back}. The only difference, compared with the case of ASEP, is that the compositions $\nu$ and $\mu$ are no longer to be understood as binary strings, but rather as strings of non-negative integers taking values in $\{0,1,\dots,r\}$.

One might then wonder how to generalize \eqref{local-mat-dual} to a multi-species setting. To address this question, we begin by elevating $\mathbb{A}$ and $\mathbb{B}$ to vector spaces, just as we did in the case of the ordinary ASEP. We define $\mathbb{A}$ to be the space of all polynomials in an infinite set of variables $\{z\}$, whose degree in the individual variable $z_i$ is bounded by $r$, for all $i \in \mathbb{Z}$. $\mathbb{B}$ is identified with the vector space $\bigotimes_{i \in \mathbb{Z}} \mathbb{C}_i^{r+1}$ with basis vectors $\bigotimes_{i \in \mathbb{Z}} \ket{\mu_i}_i$, where $\mu_i \in \{0,1,\dots,r\}$ for all $i \in \mathbb{Z}$ and where $\ket{0},\ket{1},\dots,\ket{r}$ denote the canonical basis vectors of $\mathbb{C}^{r+1}$. The operators which act on these vector spaces, $\mathbb{L}_i$ and $\mathbb{M}_i$, are essentially those of Sections \ref{sec:asep-l_i} and \ref{sec:asep-m_i}. $\mathbb{L}_i$ is defined as in \eqref{asepL}, without any modification. $\mathbb{M}_i$ is now an $(r+1)^2 \times (r+1)^2$ matrix acting in $\mathbb{C}_i^{r+1} \otimes \mathbb{C}_{i+1}^{r+1}$, with matrix entries given by \eqref{pos-back}--\eqref{diag-back}.

There is however one point of subtlety compared with the single-species ASEP: how does one choose a basis for $\mathbb{A}$, such that $\mathbb{L}_i$ acts with matrix entries that match \eqref{pos}--\eqref{diag}? This motivates the following definition:

\begin{defn}
\label{def:admiss}
Let $\nu$ denote a composition and fix a basis $\{\ket{\nu}\} = \{f_\nu(z)\}$ of $\mathbb{A}$. We say that this basis is admissible if
$\mathbb{L}_i \ket{\nu} = \sum_{\nu' \in \mathbb{A}} \ell_i(\nu',\nu) \ket{\nu'}$ for all $\nu$, where the expansion coefficients are given by \eqref{pos}--\eqref{diag}.
\end{defn}

\begin{remark}{\rm
We will say more about one possible construction of an admissible basis in the next section. It is worthwhile pointing out that the simplest basis of $\mathbb{A}$, namely $\{\ket{\nu}\} = \{\prod_{i \in \mathbb{Z}} z_i^{\nu_i}\}$, is {\it not} admissible for $r \geq 2$.
}
\end{remark}

Given an admissible basis $\{f_\nu(z)\}$ of $\mathbb{A}$, we will say that $\psi$ is a {\it local mASEP duality function} provided that, for all $i \in \mathbb{Z}$,
\begin{align}
\label{local-mat-dual2}
\mathbb{L}_i \ket{\Psi} = \mathbb{M}_i \ket{\Psi},
\quad\quad
\text{where}\ \
\ket{\Psi}
=
\sum_{\mu \in \mathbb{A}}
\sum_{\nu \in \mathbb{B}}
\psi(\nu,\mu) f_\nu(z) \ket{\mu}.
\end{align}

\section{Connection with the $t$-deformed Knizhnik--Zamolodchikov equations}
\label{se:tKZ}

This section has several aims. First, we establish a connection between the equations \eqref{local-mat-dual2} and the $t$-deformed Knizhnik--Zamolodchikov ($t$KZ) equations. More precisely, we will show that for $\psi(\nu,\mu) = \delta_{\nu,\mu}$ (trivial duality function), the equations \eqref{local-mat-dual2} are equivalent to the system of $t$KZ equations on the polynomials $\{f_{\nu}(z)\}$.

Second, we discuss how to obtain solutions of the $t$KZ equations. For this purpose, it turns out to be convenient to restrict to the space of polynomials in $n$ variables, when the number of $t$KZ equations becomes finite. In particular, we are able to make contact with a family of polynomials $\{f_{\nu}(z_1,\dots,z_n)\}$ that were considered in \cite{CantiniGW,Kasatani,KasataniT}, which have a close connection with the theory of non-symmetric Macdonald polynomials.

Third, we will outline a scheme to obtain non-trivial duality functions $\psi$ obeying \eqref{local-mat-dual2}, given a solution of the $t$KZ equations. It is based on the assumption that the polynomials $\{f_{\nu}(z)\}$ depend on an extra parameter $q$, and satisfy appropriately nice recursion relations when $q$ is specialized to certain values. In the case of the polynomials $\{f_{\nu}(z_1,\dots,z_n)\}$ studied in \cite{CantiniGW}, such recursive properties do exist, and are the subject of Sections \ref{se:rank1reduc} and \ref{se:rank2}.

\subsection{Hecke algebra, ASEP exchange relations and tKZ equations}
\label{ssec:hecke}

Consider a type $A_{n-1}$ Hecke algebra with generators $\{T_i\}_{1 \leq i \leq n-1}$, satisfying the relations
\begin{equation}
\begin{aligned}
\label{eq:Hecke}
(T_i-t)(T_i+1)=0,& \qquad T_iT_{i+1}T_i=T_{i+1}T_iT_{i+1},
\\
T_iT_j = T_j T_i,& \qquad \forall\ i,j \ \text{such that}\ \ |i-j| > 1.
\end{aligned}
\end{equation}
Both the generator $T_i$ and its inverse $T_i^{-1}$ can be realized as operators on the space of polynomials in $(z_1,\dots,z_n)$. One can easily show that
\begin{align*}
T_i
=
t
-
\left(
\frac{t z_i - z_{i+1}}{z_i - z_{i+1}}
\right)
(1-s_i),
\quad\quad
T_i^{-1}
=
t^{-1}
-
t^{-1}
\left(
\frac{t z_i - z_{i+1}}{z_i - z_{i+1}}
\right)
(1-s_i),
\end{align*}
compose as the identity, and faithfully represent the relations \eqref{eq:Hecke}.

Let $\{f_{\nu}(z)\}$ be a set of polynomials in the variables $(z_1,\dots,z_n)$, indexed by finite compositions $\nu = (\nu_1,\dots,\nu_n)$. We say that the family $\{f_{\nu}(z)\}$ is a solution of the {\it ASEP exchange relations} provided that, for all $\nu$ and
$1 \leq i \leq n-1$, the following equations hold:
\begin{equation}
\label{localtKZ}
T_i f_{(\nu_1,\dots,\nu_i,\nu_{i+1},\dots,\nu_n)}
=
\left\{
\begin{array}{rl}
f_{(\nu_1,\dots,\nu_{i+1},\nu_{i},\dots,\nu_n)},
& \quad
\nu_i >\nu_{i+1},
\\ \\
t f_{(\nu_1,\dots,\nu_{i+1},\nu_{i},\dots,\nu_n)},
& \quad
\nu_i =\nu_{i+1}.
\end{array}
\right.
\end{equation}
Note that these relations also determine $T_i f_{(\nu_1,\dots,\nu_i,\nu_{i+1},\dots,\nu_n)}$ when $\nu_i < \nu_{i+1}$. Indeed, by acting on the top equation in \eqref{localtKZ} with $T_i$ and using the quadratic relation $(T_i-t)(T_i+1)=0$, after simplification we obtain
\begin{align}
\label{other_case}
T_i f_{(\nu_1,\dots,\nu_i,\nu_{i+1},\dots,\nu_n)}
=
(t-1) f_{(\nu_1,\dots,\nu_i,\nu_{i+1},\dots,\nu_n)}
+
t f_{(\nu_1,\dots,\nu_{i+1},\nu_i,\dots,\nu_n)},
\quad
\nu_i < \nu_{i+1}.
\end{align}
Returning to the local ASEP generator \eqref{asepL}, we see that $\mathbb{L}_i=T_i-t$. Defining
\begin{align*}
\theta_i(\nu) = \left\{
\begin{array}{ll}
1, \quad & \nu_i > \nu_{i+1}, \\
0, \quad & \nu_i < \nu_{i+1}, \\
\tfrac12, & \nu_i=\nu_{i+1},
\end{array}
\right.
\qquad
\theta_i(s_i\nu) = 1-\theta_i(\nu),
\end{align*}
the relations \eqref{localtKZ} and \eqref{other_case} can collectively be written as
\begin{align}
\label{admissible_act}
\mathbb{L}_i f_{\nu}(z)  &=   t^{\theta_i(s_i\nu)}  f_{s_i \nu}(z)-  t^{\theta_i(\nu)} f_{\nu}(z) = \sum_{\nu'} \ell(\nu',\nu) f_{\nu'}(z),
\end{align}
where the coefficients in the sum are given by \eqref{pos}, \eqref{diag}. Therefore, any set of polynomials $\{f_{\nu}(z)\}$ which satisfy the exchange relations \eqref{localtKZ}, \eqref{other_case} form an admissible polynomial realization of mASEP, in the sense of Definition \ref{def:admiss}.

\begin{remark}{\rm
Restricting to compositions $\nu$ such that $\nu_i \in \{0,1\}$, one can easily show that $\{f_{\nu}(z)\} = \{\prod_{i=1}^{n} z_i^{\nu_i}\}$ is a solution of the ASEP exchange relations (indeed, this is just a rewriting of equation \eqref{action-L_i}, when it is restricted to finitely many variables).
}
\end{remark}

\begin{prop}
\label{asepdual}
The function $\psi(\nu,\mu) = \delta_{\nu,\mu}$ is a local mASEP duality function, or in other words,
\begin{align}
\label{delta}
\ket{\mathcal{I}} := \sum_{\mu} \sum_{\nu} \delta_{\nu,\mu} f_{\nu}(z) \ket{\mu} =
\sum_{\mu} f_{\mu}(z) \ket{\mu}
\quad
\text{satisfies}\ \
\mathbb{L}_i \ket{\mathcal{I}} = \mathbb{M}_i \ket{\mathcal{I}},
\quad\forall\
1 \leq i \leq n-1,
\end{align}
where $\{f_\nu(z)\}$ is a family of polynomials which satisfy the exchange relations \eqref{localtKZ} and \eqref{other_case}, $\mathbb{L}_i$ acts via \eqref{admissible_act} and $\mathbb{M}_i$ is the matrix with entries \eqref{pos-back} and \eqref{diag-back}.
\end{prop}

\begin{proof}
Writing $ \ket{\Psi} = \sum_{\mu} \sum_{\nu} \psi(\nu,\mu) f_{\nu} \ket{\mu}$, the polynomial part of the action is calculated using \eqref{admissible_act}. For any $1 \leq i \leq n-1$, we obtain
\begin{align}
\label{eq:a}
\mathbb{L}_i \ket{\Psi}
=
\sum_{\mu}
\sum_{\nu}
\psi(\nu,\mu)
\Big( t^{\theta_i(s_i\nu)}f_{s_i\nu}  - t^{\theta_i(\nu)}f_{\nu}  \Big)
\ket{\mu}
=
\sum_{\mu} \sum_{\nu} L_i\left[ \psi(\cdot,\mu)\right](\nu) f_{\nu}
\ket{\mu},
\end{align}
where in the final summation
\begin{align}
L_i\left[ \psi(\cdot,\mu)\right](\nu)=  t^{\theta_i(\nu)} \Big(\psi(s_i\nu,\mu)  - \psi(\nu,\mu)  \Big).
\label{eq:Lpsi}
\end{align}
In a similar way, the action of $\mathbb{M}_i$ gives
\begin{align}
\label{eq:b}
\mathbb{M}_i \ket{\Psi}
=
\sum_{\mu}
\sum_{\nu}
\psi(\nu,\mu)
t^{\theta_i(\mu)}
\Big( \ket{s_i \mu} - \ket{\mu} \Big)
f_{\nu}
=
\sum_{\mu}
\sum_{\nu}
M_i[\psi(\nu,\cdot)](\mu)
f_{\nu}
\ket{\mu},
\end{align}
where
\begin{align}
M_i \left[ \psi(\nu,\cdot)\right](\mu) = \Big( t^{\theta_i(s_i \mu)} \psi(\nu, s_i\mu)  -  t^{\theta_i(\mu)} \psi(\nu,\mu)  \Big).
\label{eq:Mpsi}
\end{align}
The equality of $\eqref{eq:Lpsi}$ and $\eqref{eq:Mpsi}$ is manifest when $\psi(\nu,\mu) = \delta_{\nu,\mu}$. We conclude that \eqref{eq:a} and \eqref{eq:b} are equal when $\ket{\Psi} = \ket{\mathcal{I}}$.
\end{proof}

\begin{remark}{\rm
The exchange relations \eqref{localtKZ} are also known as the {\it $t$KZ exchange equations.} They more commonly appear in the literature in terms of a stochastic higher-rank R-matrix, see \eg\ \cite{CantiniGW}.  For example, in the case $r=1$ the exchange relations \eqref{localtKZ}, and hence the duality described in Proposition~\ref{asepdual}, are recovered as the components of the equation
\begin{align}
\label{tKZ-traditional}
s_i \ket{\mathcal{I}} = \check{R}_i(z_i/z_{i+1}) \ket{\mathcal{I}},
\quad
\text{for all}\ \ i \in \mathbb{Z},
\end{align}
where $\check{R}_i(z_i/z_{i+1})$ is the R-matrix of the stochastic six-vertex model:
\be
\check{R}_i(z)=
\begin{pmatrix}
1 & 0 & 0 & 0
\\
0 & c_-(z) &  b_+(z)  & 0
\\
0 & b_-(z)  & c_+(z) & 0
\\
0 & 0 & 0 & 1
\end{pmatrix}_{i,i+1}
\ee
with
\begin{align}
b^+(z) = \displaystyle t \left( \frac{1-z}{1-tz} \right), \quad
b^-(z)  = \displaystyle \frac{1-z}{1-tz}, \quad
c^+(z) & =1-b^+(z), \quad
c^-(z) =1-b^-(z).
\end{align}
It is a simple exercise to show that \eqref{tKZ-traditional} can be cast in the form $\mathbb{L}_i \ket{\mathcal{I}} = \mathbb{M}_i \ket{\mathcal{I}}$, with $\mathbb{L}_i$ given by \eqref{asepL} and $\mathbb{M}_i$ by \eqref{asepM}. This constitutes the two equivalent forms of the $t$KZ equations, as advertised in Section \ref{se:tKZ-source}.
}
\end{remark}

In the rest of the paper we seek to go beyond the diagonal observable in Proposition \ref{asepdual}, with the aim of finding non-trivial mASEP duality functions. In order to do that, we will make contact with a particular family of polynomials $f_{\nu}$ obeying the relations \eqref{localtKZ}. This takes us on a brief detour through non-symmetric Macdonald theory.

\subsection{Non-symmetric Macdonald polynomials}

Consider polynomials in $\mathbb{C}_{q,t}[z_1,\dots,z_n]$ which are indexed by finite compositions $(\mu_1,\dots,\mu_n)$, where $t$ is (as before) related to the hopping rate in ASEP and $q$ is a new parameter. A well studied basis for $\mathbb{C}_{q,t}[z_1,\dots,z_n]$ is the basis of {\it non-symmetric Macdonald polynomials} \cite{chereda,cheredb,opdam}. Let us recall some facts about them.

Extend the Hecke algebra generated by $\{T_1,\dots,T_{n-1}\}$ and their inverses by a generator $\omega$ which acts cyclically on polynomials in $\mathbb{C}_{q,t}[z_1,\dots,z_n]$:
\begin{align}
\label{eq:omega}
(\omega g)(z_1,\ldots,z_n) &:= g(qz_n,z_1,\ldots,z_{n-1}).
\end{align}
The resulting algebraic structure is the affine Hecke algebra of type $A_{n-1}$. It has an Abelian subalgebra generated by the Cherednik--Dunkl operators $Y_i$ \cite{chered}, where
\be
Y_i := T_i\cdots T_{n-1} \omega T_{1}^{-1} \cdots T_{i-1}^{-1}.
\label{eq:Yi}
\ee
These operators mutually commute and can be jointly diagonalized. The non-symmetric Macdonald polynomials $E_{\mu} \equiv E_{\mu}(z_1,\dots,z_n;q,t)$ are the unique family of polynomials which satisfy
\begin{align}
\label{eq:monic}
E_{\mu} &= z^{\mu} + \sum_{\nu \prec \mu} c_{\mu,\nu}(q,t) z^{\nu},
\quad
c_{\mu,\nu}(q,t)
\in
\mathbb{Q}(q,t),
\\
Y_iE_{\mu} &= y_i(\mu;q,t)E_{\mu},
\quad
\forall\ 1 \leq i \leq n, \quad \mu \in \mathbb{Z}_{\geq 0}^n,
\label{eq:eigYi}
\end{align}
with eigenvalues given by
\be
\label{rhomu}
y_i(\mu;q,t)= q^{\mu_i} t^{\rho(\mu)_i+n-i+1},
\quad
\rho(\mu)= - w_{\mu}\cdot(1,2,\dots,n),
\ee
and $w_{\mu} \in S_n$ the minimal length permutation such that $\mu=w_{\mu}\cdot\mu^+$.

\begin{prop}
Let $\mu$ be any composition such that $\mu_i < \mu_{i+1}$. The non-symmetric Macdonald polynomials have the following recursive property:
\begin{align}
\label{E_exch}
E_{s_i \mu}
=
t^{-1}
\left( T_i + \frac{1-t}{1- y_{i+1}(\mu)/y_i(\mu)} \right)
E_{\mu},
\end{align}
where we abbreviate the eigenvalues \eqref{rhomu} by $y_i(\mu;q,t) \equiv y_i(\mu)$, and where we use $s_i \mu$ to denote the exchange of the parts $\mu_i$ and $\mu_{i+1}$, \ie\ $s_i \mu = (\mu_1,\dots,\mu_{i+1},\mu_i,\dots,\mu_n)$.
\end{prop}

\begin{proof}
This is a standard fact in the theory, see \cite{knop,lascouxa,lascouxb,sahi}.
\end{proof}

The non-symmetric Macdonald polynomials are meromorphic functions of the parameter $q$. Their singularities occur at points of the form $q=t^{-m}$, where $m \in \mathbb{Q}_{>0}$.\footnote{More precisely, $E_{\mu}$ may possess poles at $q=\exp(2\pi\i k/\ell) t^{-m/\ell}$ for $\ell,m \in \mathbb{N}$ and $0 \leq k \leq \ell-1$. We always focus on singular values of $q$ for which $k=0$.} These singularities play a key role in this work, so we give some results which elucidate their structure. The starting point is the following observation from \cite{Kasatani}:
\begin{prop}
\label{prop:poles_Emu}
Define a generating series $Y(w) := \sum_{i=1}^{n} Y_i w^i$ of the Cherednik--Dunkl operators, and a further generating series $y_{\mu}(w) := \sum_{i=1}^{n} y_i(\mu;q,t) w^i$ of their eigenvalues. For any composition $\mu$, we have
\begin{align}
\label{eq:poles_Emu}
E_{\mu}(z;q,t)
=
\prod_{\nu \prec \mu}
\frac{Y(w)-y_{\nu}(w)}{y_{\mu}(w)-y_{\nu}(w)}
\cdot
z^{\mu},
\end{align}
where the product is taken over all compositions $\nu$ which are smaller than $\mu$ with respect to the ordering $\prec$.
\end{prop}

\begin{proof}
By the monicity \eqref{eq:monic} of the non-symmetric Macdonald polynomials, we are able to write
\begin{align}
\label{eq:zmu}
z^{\mu} = E_{\mu} + \sum_{\nu \prec \mu} d_{\mu,\nu}(q,t) E_{\nu},
\end{align}
for some coefficients $d_{\mu,\nu}(q,t) \in \mathbb{Q}(q,t)$. We then act on this equation with the product of operators $\prod_{\nu \prec \mu} (Y(w) - y_{\nu}(w))(y_{\mu}(w)-y_{\nu}(w))^{-1}$. In view of the eigenvalue relations \eqref{eq:eigYi}, all polynomials $E_{\nu}$ with $\nu \prec \mu$ vanish under this operation, while $E_{\mu}$ is mapped to itself. Equation \eqref{eq:poles_Emu} follows immediately.
\end{proof}

\subsection{Reduction}
Although Proposition \ref{prop:poles_Emu} is easy to prove (it can be viewed as Lagrange interpolation), a slight variation of it yields an interesting statement about the structure of the singularities in $E_{\mu}$:
\begin{prop}
\label{prop:partial_E}
Fix a positive rational number $m$, a natural number $p$ and a composition $\mu$ such that
\begin{align*}
{\rm Coeff}_p[E_{\mu},m]
:=
\lim_{q \rightarrow t^{-m}} (1-qt^{m})^p E_{\mu}(z;q,t)
\end{align*}
is well defined and is non-zero. Then one has the expansion
\begin{align}
\label{gen_exp}
{\rm Coeff}_p[E_{\mu},m]
=
\lim_{q \rightarrow t^{-m}}
(1-qt^{m})^p
\left(
\sum_{\nu \in \mathcal{E}_{\mu}}
c_{\nu}(q,t) E_{\nu}(z;q,t)
\right)
\end{align}
for some family of coefficients $c_{\nu}(q,t)$, and where the sum is over the set of compositions
\begin{align}
\label{eq:specialE}
\mathcal{E}_{\mu}
=
\left\{\nu : \nu \prec \mu,\ y_{\nu}(w) = y_{\mu}(w)\ \text{at}\ q=t^{-m} \right\}.
\end{align}
\end{prop}

\begin{proof}
Start from the generic expansion \eqref{eq:zmu} and act on it with the product of operators $\prod_{\nu \prec \mu, \nu \not\in \mathcal{E}_{\mu}} (Y(w) - y_{\nu}(w))(y_{\mu}(w)-y_{\nu}(w))^{-1}$, \ie\ the same product as in the proof of Proposition \ref{prop:poles_Emu}, excluding compositions in the set $\mathcal{E}_{\mu}$. The result is the equation
\begin{align}
\label{partial_lagrange}
\prod_{\substack{\nu \prec \mu \\ \nu \not\in \mathcal{E}_{\mu}}}
\frac{Y(w)-y_{\nu}(w)}{y_{\mu}(w)-y_{\nu}(w)}
\cdot
z^{\mu}
=
E_{\mu}
+
\sum_{\nu \in \mathcal{E}_{\mu}}
d_{\mu,\nu}(q,t)
E_{\nu}.
\end{align}
Studying the left hand side of the expression \eqref{partial_lagrange}, we see that its singularities occur for compositions $\nu$ such that $y_{\mu}(w) = y_{\nu}(w)$, or more explicitly, compositions such that
\begin{align}
\label{eq:munu}
q^{\mu_i} t^{\rho(\mu)_i} = q^{\nu_i} t^{\rho(\nu)_i},
\quad
\forall\ 1\leq i \leq n.
\end{align}
For generic $q$ and $t$, it is obvious that \eqref{eq:munu} has no solution other than the tautological one, $\nu = \mu$. On the other hand, for $q=t^{-m}$ with $m \in \mathbb{Q}_{>0}$, non-trivial solutions of \eqref{eq:munu} become possible. Since we have demanded that all such compositions $\nu$ are excluded from the product, the left hand side of \eqref{partial_lagrange} has a well-defined limit when $q \rightarrow t^{-m}$. Multiplying both sides of \eqref{partial_lagrange} by $(1-qt^m)^p$ and sending $q \rightarrow t^{-m}$, the left hand side vanishes. After rearrangement, we recover \eqref{gen_exp}.
\end{proof}

The following theorem (for a special value of $p$) is a stronger version of Proposition \ref{prop:partial_E}, in which only a single composition in the sum \eqref{gen_exp} is retained. We were unable to locate this result anywhere in the literature.

\begin{thm}
\label{thm_egr_minor}
Fix $m,p,\mu$ as in the statement of Proposition \ref{prop:partial_E}, and assume in addition that $p=|\mathcal{E}_\mu|$, where $\mathcal{E}_\mu$ is defined in \eqref{eq:specialE}. Then there exists a unique composition $\nu$ for which
\begin{align*}
E_{\nu}(z;t^{-m},t)
:=
\lim_{q \rightarrow t^{-m}} E_{\nu}(z;q,t)
\end{align*}
is well defined and such that
\begin{align}
\label{res_thm}
{\rm Coeff}_p [E_\mu,m] \propto E_{\nu}(z;t^{-m},t).
\end{align}
\end{thm}

\begin{proof}
We start from the expression \eqref{eq:poles_Emu} for $E_{\mu}$ and assume there are exactly $p$ solutions of \eqref{eq:munu}, meaning that the cardinality of $\mathcal{E}_{\mu}$ is equal to $p$. Call these solutions $\nu[1],\dots,\nu[p]$ and assume that they have the ordering $\nu[1] \prec \cdots \prec \nu[p]$. Then by direct calculation on \eqref{eq:poles_Emu}, we have
\begin{align}
\label{eq:chain}
{\rm Coeff}_p[E_{\mu},m]
\propto
\left[
\prod_{\substack{\kappa \prec \mu \\ \kappa \not\in \mathcal{E}_{\mu}}}
\frac{Y(w)-y_{\kappa}(w)}{y_{\mu}(w)-y_{\kappa}(w)}
\cdot
\prod_{i=1}^{p}
(Y(w) - y_{\nu[i]}(w))
\cdot
z^{\mu}
\right]_{q=t^{-m}}
\end{align}
where we suppress the proportionality factors which arise in taking this limit. There cannot be any singularities on the right hand side of \eqref{eq:chain}, since $\nu[1],\dots,\nu[p]$ are the only compositions for which \eqref{eq:munu} holds, so the specialization $q=t^{-m}$ can be freely taken.

For generic $q$, it is an easy consequence of \eqref{eq:monic}, \eqref{eq:eigYi} and \eqref{eq:zmu} in combination that
\begin{align}
\label{eq:proj}
(Y(w)-y_{\mu}(w)) z^{\mu}
=
\sum_{\nu \prec \mu}
e_{\mu,\nu}(q,t;w) z^{\nu},
\end{align}
where the sum on the right hand side is over compositions $\nu$ which are strictly less than $\mu$ with respect to the ordering $\prec$, for some coefficients $e_{\mu,\nu}(q,t;w)$ which are polynomial in $q$. The polynomiality of the coefficients is ensured by \eqref{eq:omega} and \eqref{eq:Yi}. This equation therefore extends to specializations $q=t^{-m}$. Equation \eqref{eq:chain} can now be further simplified, by the following iterative procedure. Since $ y_{\nu[p]}(w) = y_\mu(w)$ at $q=t^{-m}$, by repeated use of \eqref{eq:proj} we see that
\begin{align*}
\left[
\prod_{\substack{\nu[p] \prec \kappa \prec \mu }}
\frac{Y(w)-y_{\kappa}(w)}{y_{\mu}(w)-y_{\kappa}(w)}
\cdot
(Y(w) - y_{\nu[p]}(w))
\cdot
z^{\mu}
\right]_{q=t^{-m}}
\propto
\left(
z^{\nu[p]}
+
\sum_{\nu \prec \nu[p]}
g_{\nu}(t;w)
z^{\nu}\right),
\end{align*}
for appropriate coefficients $g_{\nu}(t;w)$; \ie\ starting from the monomial $z^{\mu}$, it is successively lowered to monomials $z^{\kappa}$ which are smaller in the $\prec$ ordering, until we arrive at $z^{\nu[p]}$. We can then repeat this process, using the fact that $y_{\nu[i-1]}(w) = y_{\nu[i]}(w)$ at $q=t^{-m}$, for all $1 < i \leq p$. We arrive ultimately at the expression
\begin{align*}
{\rm Coeff}_p[E_{\mu},m]
\propto
\left[
\prod_{\substack{\kappa \prec \nu[1]}}
\frac{Y(w)-y_{\kappa}(w)}{y_{\mu}(w)-y_{\kappa}(w)}
\cdot
\left( z^{\nu[1]} + \sum_{\nu \prec \nu[1]} h_{\nu}(t;w) z^{\nu} \right)
\right]_{q=t^{-m}}
\end{align*}
for some coefficients $h_{\nu}(t;w)$, and note that all sub-leading terms in the sum vanish under the product of operators, by exactly the same filtering argument used above. We have thus shown that
\begin{align*}
{\rm Coeff}_p[E_{\mu},m]
\propto
\left[
\prod_{\substack{\kappa \prec \nu[1]}}
\frac{Y(w)-y_{\kappa}(w)}{y_{\nu[1]}(w)-y_{\kappa}(w)}
\cdot
z^{\nu[1]}
\right]_{q=t^{-m}}
=
E_{\nu[1]}(z;t^{-m},t),
\end{align*}
establishing both the existence and uniqueness claim.
\end{proof}
Notice that this procedure specifies the $\nu$ appearing in \eqref{res_thm} as the minimal composition (with respect to $\prec$) which satisfies \eqref{eq:munu} at $q=t^{-m}$. It does not, however, give $\nu$ constructively: one still needs to do the work of finding solutions of \eqref{eq:munu}.

Based on experimentation with the non-symmetric Macdonald polynomials we are led to make the following conjecture, generalizing Theorem~\ref{thm_egr_minor} to arbitrary values of $p$, which we were unable to prove in full generality. All of our subsequent results on duality functions can be (and are) proved independently of this conjecture, but it remains an important conceptual (if not technical) cornerstone of this work:

\begin{conj}
\label{theorema_egregium}
Fix a positive rational number $m$, a natural number $p$ and a composition $\mu$ such that ${\rm Coeff}_p[E_{\mu},m]$ is well defined and non-zero. Then there exists a unique composition $\nu$ for which
\begin{align*}
E_{\nu}(z;t^{-m},t)
:=
\lim_{q \rightarrow t^{-m}} E_{\nu}(z;q,t)
\end{align*}
is well defined and such that
\begin{align}
{\rm Coeff}_p [E_\mu,m] \propto E_{\nu}(z;t^{-m},t).
\end{align}
\end{conj}

\subsection{Another non-symmetric basis}

In this work we make use of a further set of non-symmetric polynomials, which also comprise a basis of $\mathbb{C}_{q,t}[z_1,\dots,z_n]$. We refer to them as {\it ASEP polynomials,} and denote them by $f_{\mu} = f_{\mu}(z_1,\dots,z_n;q,t)$. They are defined as the unique family of polynomials which satisfy
\begin{align}
\label{f_init}
f_{\delta}(z;q,t) = E_{\delta}(z;q,t),
\quad \forall\ \delta = (\delta_1 \leq \cdots \leq \delta_n),
\\
f_{s_i \mu}(z;q,t) = T^{-1}_i f_{\mu}(z;q,t),
\quad \text{when}\ \ \mu_i < \mu_{i+1},
\label{f_exch}
\end{align}
where, as before, $s_i \mu = (\mu_1,\dots,\mu_{i+1},\mu_i,\dots,\mu_n)$. Clearly by repeated use of \eqref{f_exch}, one is able to construct $f_{\mu}$ for any composition, starting from $f_{\mu^{-}} = E_{\mu^{-}}$. Furthermore, because of the Hecke algebra relations \eqref{eq:Hecke}, $f_{\mu}$ is independent of the order in which one performs the operations \eqref{f_exch}, making the definition unambiguous.

It can be shown \cite{CantiniGW,KasataniT} that the ASEP polynomials are equivalently defined as the unique monic polynomials $f_{\mu} = z^{\mu} + \sum_{\nu \prec \mu} c_{\mu,\nu}(q,t) z^{\nu}$, for some family of coefficients $c_{\mu,\nu}(q,t)$, satisfying the $t$KZ relations \eqref{localtKZ} for $1 \leq i \leq n-1$, and the cyclic boundary condition
\be
\label{cyclic}
f_{\mu_n,\mu_1,\dots,\mu_{n-1}}(qz_n,z_1,\ldots,z_{n-1};q,t) = q^{\mu_n}f_{\mu_1,\dots,\mu_n}(z_1,\dots,z_n;q,t).
\ee
In view of the discussion in Section \ref{ssec:hecke}, they are therefore fundamental in the study of duality functions for the mASEP. This is not the first time that the family $\{f_{\mu}\}$ has appeared in the context of stochastic processes: in \cite{CantiniGW} these polynomials also played the role of (inhomogeneous generalizations of) stationary state probabilities in the mASEP on a ring.

We stress that, in general, $f_{\mu} \not= E_{\mu}$; the non-symmetric Macdonald and ASEP polynomials coincide when their indexing composition is an anti-partition, but are otherwise different, which is readily apparent from their different recursive properties \eqref{E_exch} and \eqref{f_exch}. One basis can be expanded triangularly in terms of the other, however, as we now show:

\begin{defn}
A {\it composition sector} is the set of all compositions with a common anti-dominant (or dominant) ordering. If $\mu$ is a composition, the composition sector $\sigma(\mu)$ is the following set:
\begin{align*}
\sigma(\mu) := \{\nu | \nu^{-} = \mu^{-}\}.
\end{align*}
\end{defn}

\begin{prop}
For any composition $\mu$, there are unique triangular expansions
\begin{align}
\label{E_expand}
E_{\mu}(z;q,t)
&=
f_{\mu}(z;q,t)
+
\sum_{\substack{\nu \in \sigma(\mu) \\ \nu \prec \mu}}
c_{\mu,\nu}(q,t)
f_{\nu}(z;q,t),
\\
\label{f_expand}
f_{\mu}(z;q,t)
&=
E_{\mu}(z;q,t)
+
\sum_{\substack{\nu \in \sigma(\mu) \\ \nu \prec \mu}}
d_{\mu,\nu}(q,t)
E_{\nu}(z;q,t),
\end{align}
for some coefficients $c_{\mu,\nu}(q,t)$ and $d_{\mu,\nu}(q,t)$, relating the non-symmetric Macdonald and ASEP bases.
\end{prop}

\begin{proof}
The uniqueness claim is immediate, since both families are bases for $\mathbb{C}_{q,t}[z_1,\dots,z_n;q,t]$. To prove the form of the expansion \eqref{E_expand}, we note that it holds trivially in the case where $\mu$ is an anti-partition. Based on this, assume that it holds for some composition $\mu$ such that $\mu_i < \mu_{i+1}$, for some $1 \leq i \leq n-1$. By application of \eqref{E_exch}, we then have
\begin{align}
\label{brute_force}
E_{s_i\mu}
&=
t^{-1}
\left( T_i + \frac{1-t}{1-y_{i+1}(\mu)/y_i(\mu)} \right)
\left(
f_{\mu}
+
\sum_{\substack{\nu \in \sigma(\mu) \\ \nu \prec \mu}}
c_{\mu,\nu}(q,t)
f_{\nu}
\right).
\end{align}
We need to act with the Hecke generator $T_i$ on the sum over ASEP polynomials. The action of $T_i$ on any given $f_{\nu}$ produces some linear combination of $f_{\nu}$ and $f_{s_i\nu}$, as can be seen from \eqref{localtKZ} and \eqref{other_case}. Both $f_{\nu}$ and $f_{s_i\nu}$ obviously lie in the composition sector $\sigma(\mu) \equiv \sigma(s_i \mu)$. Now when $\mu_i < \mu_{i+1}$ and $\nu \prec \mu$ hold, it is clear that both $\nu \prec s_i \mu$ and $s_i \nu \prec s_i \mu$ also hold. Using these observations in \eqref{brute_force}, we can then write
\begin{align*}
E_{s_i\mu}
&=
f_{s_i \mu}
+
\sum_{\substack{\nu \in \sigma(s_i \mu) \\ \nu \prec s_i \mu}}
c_{s_i\mu,\nu}(q,t)
f_{\nu}
\end{align*}
for appropriate coefficients $c_{s_i\mu,\nu}(q,t)$. Note that the coefficient of $f_{s_i \mu}$ must be $1$, using equation \eqref{other_case} to calculate $t^{-1} T_i f_{\mu}$. This proves that \eqref{E_expand} holds generally, by induction.

Finally, by virtue of \eqref{E_expand}, the transition matrix $c$ with entries $c_{\mu,\nu}(q,t)$ is block-diagonal over composition sectors, with triangular blocks. It can therefore be inverted to yield \eqref{f_expand}, where the transition matrix $d$ with entries $d_{\mu,\nu}(q,t)$ is the inverse of $c$.
\end{proof}

Like the non-symmetric Macdonald polynomials, the ASEP polynomials may become singular when $q=t^{-m}$,
$m \in \mathbb{Q}_{>0}$. To clarify the structure of these singularities, we seek a result which directly parallels Conjecture~\ref{theorema_egregium}.
\begin{thm}
\label{thm:sector}
Fix a positive rational number $m$, a natural number $p$ and an anti-partition $\delta$ for which Conjecture~\ref{theorema_egregium} holds. Then there exists a unique anti-partition $\epsilon$ such that
\begin{align*}
f_{\nu}(z;t^{-m},t)
:=
\lim_{q \rightarrow t^{-m}}
f_{\nu}(z;q,t)
\end{align*}
is well defined for all compositions $\nu \in \sigma(\epsilon)$, and such that
\begin{align}
\label{unique_sector}
{\rm Coeff}_p[f_{\mu},m]
=
\sum_{\nu \in \sigma(\epsilon)}
\psi(\nu,\mu;t)
f_{\nu}(z;t^{-m},t),
\end{align}
for all $\mu \in \sigma(\delta)$ and suitable coefficients $\psi(\nu,\mu;t)$.
\end{thm}

\begin{proof}
Let us begin by analyzing the case where $\mu=\delta$. In that case, using the direct equivalence of ASEP and non-symmetric Macdonald polynomials and the result of Conjecture~\ref{theorema_egregium} we have
\begin{align*}
{\rm Coeff}_p[f_{\delta},m]
\equiv
{\rm Coeff}_p[E_{\delta},m]
\propto
E_{\kappa}(z;t^{-m},t),
\end{align*}
where $\kappa$ is the minimal composition satisfying the relations $y_i(\delta) = y_i(\kappa)$ at $q=t^{-m}$. Let $\epsilon = \kappa^{-}$. Using equation \eqref{E_expand}, we know that an expansion of the form
\begin{align*}
E_{\kappa}(z;q,t)
&=
f_{\kappa}(z;q,t)
+
\sum_{\substack{\nu \in \sigma(\epsilon) \\ \nu \prec \kappa}}
d_{\kappa,\nu}(q,t)
f_{\nu}(z;q,t)
\end{align*}
exists, and each $f_{\nu}$ appearing on the right hand is relatable to $f_{\epsilon} = E_{\epsilon}$ by a successive action of inverse Hecke generators $T^{-1}_i$. The action of such generators does not introduce any singular points in $q$, and we know that $\lim_{q \rightarrow t^{-m}} E_{\epsilon}$ is well defined; it follows that one can freely set $q=t^{-m}$ in the above equation, establishing that
\begin{align}
\label{eq:anti_claim}
{\rm Coeff}_p[f_{\delta},m]
\propto
f_{\kappa}(z;t^{-m},t)
+
\sum_{\substack{\nu \in \sigma(\epsilon) \\ \nu \prec \kappa}}
d_{\kappa,\nu}(t^{-m},t)
f_{\nu}(z;t^{-m},t).
\end{align}
This proves the claim \eqref{unique_sector} for anti-partitions $\mu=\delta$. The general $\mu$ case now follows immediately, by acting on the equation \eqref{eq:anti_claim} with products of inverse Hecke generators. This is permitted, since (as before) the action of these generators commutes with the limits being taken, and it allows $f_{\delta}$ to be converted into an arbitrary ASEP polynomial $f_{\mu}$. The action of $T_i^{-1}$ on the right hand side of \eqref{eq:anti_claim} also manifestly preserves the sector being summed over.

\end{proof}

\subsection{Dualities from reductions of ASEP polynomials}

In the previous sections we have outlined some of the theory surrounding the non-symmetric Macdonald and ASEP polynomials, with particular emphasis on their singular points in the parameter $q$. We now apply these results to the construction of non-trivial duality functions in mASEP systems. The following result is the central idea of this paper:

\begin{thm}
\label{thm:central}
Fix a positive rational number $m$, a natural number $p$ and an anti-partition $\delta$ such that for all compositions $\mu \in \sigma(\delta)$ there exists an expansion
\begin{align}
\label{f_sum}
{\rm Coeff}_p[f_{\mu},m]
=
\sum_{\nu \in \sigma(\epsilon)}
\psi(\nu,\mu;t)
f_{\nu}(z;t^{-m},t),
\end{align}
where $\epsilon$ is some other known anti-partition.\footnote{The expansion \eqref{f_sum} is guaranteed to be possible if the conditions in Theorem \ref{thm:sector} are met, namely the validity of Conjecture \ref{theorema_egregium} However it is sometimes possible to show that \eqref{f_sum} holds, independently of Conjecture \ref{theorema_egregium}, by proceeding via the weaker Proposition \ref{prop:partial_E}. This is the course of action that we take in Sections \ref{se:rank1reduc} and \ref{se:rank2}.} Then $\psi(\nu,\mu;t) \equiv \psi(\nu,\mu)$ defines a local duality function of the mASEP with generator $L_i$ given by \eqref{Mi}--\eqref{diag}, and the mASEP with generator $M_i$ given by \eqref{Li}--\eqref{diag-back}. Explicitly, we have
\begin{align}
\label{functional_dual}
L_i [\psi(\cdot,\mu)](\nu) = M_i [\psi(\nu,\cdot)](\mu), \qquad\forall\ 1 \leq i \leq n-1,
\end{align}
where the left hand side of \eqref{functional_dual} is given by \eqref{eq:Lpsi}, and the right hand side by \eqref{eq:Mpsi}.
\end{thm}

\begin{proof}
From Proposition \ref{asepdual}, we know that
\begin{align*}
\ket{\mathcal{I}}
=
\sum_{\mu \in \sigma(\delta)}
f_{\mu}(z;q,t)
\ket{\mu}
\end{align*}
satisfies $\mathbb{L}_i \ket{\mathcal{I}} = \mathbb{M}_i \ket{\mathcal{I}}$ for all $1 \leq i \leq n-1$. Exploiting the freedom to take limits of $q$, since it does not appear in the local mASEP generators, we see that
\begin{align*}
\ket{\mathcal{I}_{p,m}}
:=
{\rm Coeff}_p[\ket{\mathcal{I}},m]
=
\sum_{\mu \in \sigma(\delta)}
\sum_{\nu \in \sigma(\epsilon)}
\psi(\nu,\mu;t)
f_\nu(z;t^{-m},t)
\ket{\mu}
\end{align*}
satisfies $\mathbb{L}_i \ket{\mathcal{I}_{p,m}} = \mathbb{M}_i \ket{\mathcal{I}_{p,m}}$ for all $1 \leq i \leq n-1$. Converting this to its functional form, we obtain precisely the relations \eqref{functional_dual}.
\end{proof}

\begin{remark}{\rm
The anti-partitions $\delta$ and $\epsilon$ label the particle content of the two mASEP systems appearing in Theorem \ref{thm:central}. More precisely, Theorem \ref{thm:central} presents a duality between one mASEP with $m_i(\delta)$ particles of type $i$ and another mASEP with $m_i(\epsilon)$ particles of type $i$, $0 \leq i \leq r$.
}
\end{remark}

\begin{remark}{\rm
Theorem \ref{thm:central} gives rise to a diverse collection of duality functions. Once the particle content of one mASEP system is fixed by choosing $\delta$, there will in general be multiple choices of $m \in \mathbb{N}$ and $p \in \mathbb{Q}_{>0}$ for which ${\rm Coeff}_p[f_{\delta},m]$ exists and is non-zero. Each such choice will give rise to a different $\epsilon$, labelling the particle content of the second, reduced mASEP system.

It is beyond the scope of the present paper to explore all possible duality functions arising from Theorem \ref{thm:central}. One of the obstacles of such a classification is that one needs a way of calculating the coefficients appearing in \eqref{f_sum}, which is difficult in full generality. We hope to return to this problem in a future publication.

For the purposes of the current work, we prefer to analyse \eqref{f_sum} for some special choices of $\{\delta,p,m\}$. Section \ref{se:rank1reduc} will look at the case $\{\delta,p,m\} = \{(0^{n-m}, r^{m}),1,m\}$ for general $r \geq 1$. Section \ref{se:rank2} deals with the case $\{\delta,p,m\} = \{(0^{n-m_1-m_2}, 1^{m_1}, 2^{m_2}),1,M\}$ for general $m_1,m_2,M \geq 1$.
}
\end{remark}

\section{Explicit formulae for the ASEP polynomials}

In order to calculate expansions of the form \eqref{f_sum} explicitly, it naturally helpful to have explicit expressions for the polynomials $f_{\mu}(z;q,t)$ themselves. Such formulae were obtained in \cite{CantiniGW,GierW}, and turn out to be quite expedient for the purposes of this paper, since they lare bare the structure of the singularities of $f_{\mu}(z;q,t)$ as a function of $q$.

\subsection{Matrix product formula for $f_{\mu}(z;q,t)$}
\label{se:mpa}

Let us recall some of the details of the matrix product Ansatz. Given a composition $\mu$ whose largest part is equal to $r$, one seeks a construction of the form
\begin{align}
\label{general_mpa}
f_{\mu}(z_1,\dots,z_n;q,t)
=
\Omega_{\mu}(q,t)
\times
{\rm Tr}\Big(
A_{\mu_1}(z_1) \dots A_{\mu_n}(z_n) S
\Big),
\end{align}
where $\{A_i(z)\}_{0 \leq i \leq r}$ and $S$ are a collection of explicit matrices, and $\Omega_{\mu}$ is a normalization constant (recall that $f_{\mu}$ is monic, \ie\ it expands as $f_{\mu} = z^{\mu} + \sum_{\nu \prec \mu} c_{\mu,\nu}(q,t) z^{\nu}$). To proceed with the construction \eqref{general_mpa}, two steps are necessary. First, one needs to translate the exchange relations \eqref{localtKZ} and \eqref{cyclic}, which uniquely characterize the family $\{f_{\mu}\}$, into algebraic relations between the $A_i(z)$ and $S$ operators. The algebraic structure which arises from this is the {\it Zamolodchikov--Faddeev (ZF) algebra.}\footnote{In fact the resulting structure is an extended version of the ZF algebra, since it not only prescribes commutation relations between the operators $\{A_i(z)\}$, but also with the ``twist'' operator $S$.} Second, one needs to seek a suitable representation of this algebra, so that the trace in \eqref{general_mpa} can be taken.

Following these steps, an explicit matrix product expression \eqref{general_mpa} for $f_{\mu}(z;q,t)$ was obtained in \cite{CantiniGW}. It involves a family of infinite-dimensional matrices $\phi,\phid,k$ which satisfy the $t$-boson algebra. Their matrix entries are given explicitly by
\begin{align*}
[\phi]_{i,j}
=
\delta_{i+1,j} (1-t^i),
\quad\quad
[\phid]_{i,j}
=
\delta_{i,j+1},
\quad\quad
[k]_{i,j}
=
\delta_{i,j} t^i,
\quad
\text{for all}\
i,j \in \mathbb{N}.
\end{align*}
It is easy to check that this provides a faithful representation of the $t$-boson algebra $\mathfrak{B}$, \ie\ the matrices obey the relations
\begin{align*}
\phi \phid - t \phid \phi = 1-t,
\quad\quad
t k \phi = \phi k,
\quad\quad
k \phid = t \phid k.
\end{align*}
We refer the reader to \cite{CantiniGW} for the matrix product formula for generic $f_{\mu}(z;q,t)$. In this paper we focus on two sub-families of compositions for which the formula \eqref{general_mpa} becomes simple. We detail these below:

\subsubsection{The case $\mu^{-} = (0^{n-m}, r^m)$.} We begin by analyzing the matrix product expression when $\mu$ is a composition with parts of size $0$ and size $r$, only. Let $L(z)$ denote the following $2 \times 2$ matrix, whose entries are $t$-bosons:
\begin{align*}
L(z) =
\left(
\begin{array}{cc}
1 & \phi
\\
z \phid & z
\end{array}
\right),
\end{align*}
\ie\ the entries of $L(z)$ are themselves to be understood as infinite dimensional matrices. From this, construct a two-component vector
\begin{align}
\label{a0_ar}
\begin{pmatrix} A_0(z) \\ A_r(z) \end{pmatrix}
:=
\underbrace{L(z) \stackrel{.}{\otimes} \cdots \stackrel{.}{\otimes} L(z)}_{r-1}
\begin{pmatrix} 1 \\ z \end{pmatrix},
\end{align}
where $L(z)$ is composed $r-1$ times under the operation $\stackrel{.}{\otimes}$, meaning matrix multiplication combined with taking Kronecker products of matrix entries:
\begin{align*}
\begin{pmatrix}
a & b
\\
c & d
\end{pmatrix}
\stackrel{.}{\otimes}
\begin{pmatrix}
e & f
\\
g & h
\end{pmatrix}
:=
\begin{pmatrix}
a \otimes e + b \otimes g & a \otimes f + b \otimes h
\\
c \otimes e + d \otimes g& c \otimes f + d \otimes h
\end{pmatrix}.
\end{align*}
The resulting operators $A_0(z)$ and $A_r(z)$ are thus polynomial in $z$, with coefficients in $\mathfrak{B}^{\otimes_{r-1}}$. One can easily calculate the first few examples of these operators:
\begin{align*}
r = 1: \quad & A_0(z) = 1, \quad A_1(z) = z
\\
r=2: \quad & A_0(z) = 1 + z \phi, \quad A_2(z) = z \phid + z^2
\\
r=3: \quad  &A_0(z) = 1 \otimes 1 + z (1 \otimes \phi + \phi \otimes \phid) +
z^2 (\phi \otimes 1),
\\
& A_3(z) = z (\phid \otimes 1) + z^2 (\phid \otimes \phi + 1 \otimes \phid) + z^3 (1 \otimes 1).
\end{align*}

\begin{prop}
Let $\mu$ be a composition with anti-dominant ordering $\mu^{-} = (0^{n-m},r^m)$. Then
\begin{align}
\label{rankr_mpa}
f_{\mu}(z_1,\dots,z_n;q,t)
=
\prod_{i=1}^{r-1} (1-q^i)
\times
{\rm Tr}
\left(
A_{\mu_1}(z_1) A_{\mu_2}(z_2) \dots A_{\mu_n}(z_n)
k^{u(r-1)} \otimes k^{u(r-2)} \otimes \cdots \otimes k^{u}
\right),
\end{align}
where each operator $A_i(z)$ is given by \eqref{a0_ar}, $q$ is parametrized through $u$ via $q := t^{u}$, and the trace is taken over $\mathfrak{B}^{\otimes_{r-1}}$ and is to be understood as a formal power series in $t$.
\end{prop}

\begin{proof}
This follows from the matrix product expression in \cite{CantiniGW}, under some simplifications. The result in \cite{CantiniGW} applies to generic compositions $\mu$, and makes use of $r$ commuting copies of the $t$-boson algebra $\{\mathfrak{B}_i\}_{1 \leq i \leq r}$, where $r$ is the largest part of $\mu$. However, whenever $\mu$ consists of less than $r$ distinct non-zero parts, the dependence on some of these families drops out. In the case at hand, $\mu$ consists of only one type of non-zero part (namely, $r$), and can therefore be expressed via a matrix product that only uses a single copy of $\mathfrak{B}$. It is this simplification of the formula in \cite{CantiniGW} which gives rise to \eqref{rankr_mpa}; for simplicity we will suppress further details.
\end{proof}

\begin{remark}{\rm
One can use equation \eqref{rankr_mpa} to obtain a completely explicit expression for any given polynomial $f_{\mu}(z_1,\dots,z_n;q,t)$, where $\mu^{-} = (0^{n-m},r^m)$. The calculation of the trace amounts to taking geometric series, and for that reason $f_{\mu}$ acquires denominators of the form $(1-q^i t^j)$. This is in accordance with the singularities that $f_{\mu}$ is expected to have, as a function of $q$.
}
\end{remark}

\subsubsection{The case $\mu^{-} = (0^{n-m_1-m_2}, 1^{m_1}, 2^{m_2})$.} An even simpler case is that of compositions whose parts are of size $2$, or less. We refer to these as {\it rank-two compositions.} In that situation we define directly
\begin{align}
\label{ZF-ops}
A_0(z) = 1+z \phi,
\quad
A_1(z) = z k,
\quad
A_2(z) = z \phid + z^2.
\end{align}

\begin{prop}
For any rank-two composition $\mu$, we have
\begin{align}
\label{MPA}
f_{\mu}(z_1,\dots,z_n;q,t)
=
\left(1-q t^{m_1}\right)
\times
{\rm Tr}
\Big(
A_{\mu_1}(z_1)
\dots
A_{\mu_n}(z_n)
k^{u}
\Big),
\end{align}
where $m_1=m_1(\mu)$ is the number of parts in $\mu$ equal to 1, $q = t^u$, and where the trace is again to be understood as a formal power series in $t$.
\end{prop}

\begin{proof}
This is exactly the special case $r=2$ of the matrix product formula in \cite{CantiniGW}; see Section 3 therein.
\end{proof}

\subsection{Summation formulae}

In \cite{GierW} an alternative formula for $f_{\mu}(z;q,t)$ was obtained, in terms of multiple summations over the symmetric group $S_n$. This expression can be derived from the matrix product formula of \cite{CantiniGW}, by explicitly evaluating all traces which appear. In view of its complexity we do not repeat the general formula here, but again focus on the special cases which are of interest in this paper.

\subsubsection{The case $\delta = (0^{n-m}, r^m)$.} Let $\alpha$ and $\beta$ be rank-one compositions, and for any $j \geq 1$ define coefficients
\begin{align*}
C_j(\alpha,\beta;q,t)
:=
{\rm Tr}\left( L(\alpha_1,\beta_1) \dots L(\alpha_n,\beta_n) k^{ju} \right),
\quad
\text{where}\
L(\alpha,\alpha) = 1,
\quad
L(0,1) = \phi,
\quad
L(1,0) = \phid.
\end{align*}
These coefficients are rational functions in $q=t^u$ and $t$; for given rank-one compositions $\alpha$ and $\beta$ they can be readily evaluated by tracing over the resulting product of infinite-dimensional matrices. We will make use of the following key properties:
\begin{prop}
\label{prop:coeff}
$C_j(\alpha,\beta;q,t)$ vanishes unless $|\alpha| = |\beta|$. In the case where 
$\#\{(\alpha_i,\beta_i) = (0,1)\} = \#\{(\alpha_i,\beta_i) = (1,0)\} = m$, one has
\begin{align}
\label{c_j}
C_j(\alpha,\beta;q,t)
=
\frac{p_j(\alpha,\beta;q,t)}{\prod_{i=0}^{m} (1-q^j t^i)},
\end{align}
where $p_j(\alpha,\beta;q,t)$ is polynomial in $(q,t)$.
\end{prop}

\begin{prop}
Fix an anti-partition $\delta = (0^{n-m},r^m)$ and a corresponding projection onto rank-one, $\delta^{*} = (0^{n-m},1^m)$. The formula
\begin{align}
\label{sum-rankr}
f_{\delta}
=
\prod_{i=1}^{r-1} (1-q^i)
\times
\sum_{\mu[1] \in \sigma(\delta^{*})}
\cdots
\sum_{\mu[r-1] \in \sigma(\delta^{*})}
z^{\delta^{*}}
\left(
\prod_{j=1}^{r-1}
C_{j}\Big(\mu[j+1],\mu[j];q,t\Big)
z^{\mu[j]}
\right)
\end{align}
holds, where $\mu[1],\dots,\mu[r-1]$ are dummy indices, each being summed over all rank-one compositions in the sector $\sigma(\delta^{*})$, and $\mu[r] \equiv \delta^{*}$.
\end{prop}

\begin{proof}
This follows from the matrix product formula \eqref{rankr_mpa}, by decomposing the trace over the $r-1$ factors in the tensor product, and using the definition \eqref{a0_ar} of the $A_i(z)$ operators.
\end{proof}

\subsubsection{The case $\delta = (0^{n-m_1-m_2}, 1^{m_1}, 2^{m_2})$.}

\begin{prop}
Fix a rank-two anti-partition $\delta = (0^{n-m_1-m_2}, 1^{m_1}, 2^{m_2})$. The formula
\begin{align*}
f_{\delta}
=
\prod_{j=1}^{m_1+m_2}
(z_{n-j+1})
\times
\sum_{i=0}^{m_2}
t^{i m_1}
\prod_{j=1}^{i}
\left(
\frac{1-t^j}{1-qt^{m_1+j}}
\right)
e_i\Big(z_1,\dots,z_{n-m_1-m_2}\Big)
e_{m_2-i}\Big(z_{n-m_2+1},\dots,z_n\Big)
\end{align*}
holds, where $e_i$ denotes the $i$-th elementary symmetric polynomial, given by the generating series expression
\begin{align*}
\sum_{i=0}^N
e_i(x_1,\dots,x_N) y^i
=
\prod_{j=1}^{N}
(1+x_j y),
\quad
\text{for any alphabet}\ (x_1,\dots,x_N). 
\end{align*}
\end{prop}

\begin{proof}
Using the matrix product formula \eqref{MPA} in the case $\mu = (0^{n-m_1-m_2}, 1^{m_1}, 2^{m_2})$, we find that
\begin{align*}
f_{\delta}
& =
(1-qt^{m_1}) \times {\rm Tr}
\left(
\prod_{i=1}^{n-m_1-m_2}
(1+z_i \phi)
\cdot
\prod_{j=n-m_1-m_2+1}^{n-m_2}
(z_j k)
\cdot
\prod_{l=n-m_2+1}^{n}
(z_l \phid + z_l^2)
\cdot
k^u
\right)
\\
& =
(1-qt^{m_1}) 
\prod_{j=1}^{m_1+m_2}
(z_{n-j+1})
 \times {\rm Tr}
\left(
\prod_{i=1}^{n-m_1-m_2}
(1+z_i \phi)
\cdot
\prod_{l=n-m_2+1}^{n}
(t^{m_1} \phid + z_l)
\cdot
k^{u+m_1}
\right),
\end{align*}
where we have used the commutation relation $k \phid = t \phid k$ to bring the product $k^{m_1}$ from the middle to the right of the expression. One can now evaluate the trace directly; the only terms which will have a non-zero trace are those proportional to $\phi^a \phi^{\dagger a}$, where $0 \leq a \leq m_2$. Summing over all such possibilities, we immediately find that
\begin{multline}
\label{sum-2}
f_{\delta}
=
(1-qt^{m_1}) 
\prod_{j=1}^{m_1+m_2}
(z_{n-j+1})
\times
\\
\sum_{a=0}^{m_2}
t^{a m_1}
{\rm Tr}
\left(
\phi^a \phi^{\dagger a}
k^{u+m_1}
\right)
e_a\Big(z_1,\dots,z_{n-m_1-m_2}\Big)
e_{m_2-a}\Big(z_{n-m_2+1},\dots,z_n\Big).
\end{multline}
Finally, the trace in \eqref{sum-2} can be evaluated explicitly:
\begin{align*}
{\rm Tr}
\left(
\phi^a \phi^{\dagger a}
k^{u+m_1}
\right)
=
\frac{1}{1-t^{u+m_1}}
\prod_{i=1}^{a}
\left(
\frac{1-t^i}{1-t^{u+m_1+i}}
\right)
=
\frac{1}{1-q t^{m_1}}
\prod_{i=1}^{a}
\left(
\frac{1-t^i}{1-q t^{m_1+i}}
\right),
\end{align*}
under the identification $t^u \equiv q$. Substituting this into \eqref{sum-2} yields the desired result.
\end{proof}

\section{Rank-one ASEP dualities}
\label{se:rank1reduc}

In this section we show how certain self-dualities between asymmetric simple exclusion processes, first found in \cite{schuetz97} and later elaborated in terms of ASEP generators in \cite{bcs}, arise within our formalism. This is achieved in three steps: {\bf 1.} The identification of suitable sectors $\delta$ and $\epsilon$ for the use of Theorem \ref{thm:central}; {\bf 2.} The calculation of the coefficients $\psi(\nu,\mu;t)$ in \eqref{f_sum} for all $\mu \in \sigma(\delta)$ and $\nu \in \sigma(\epsilon)$; {\bf 3.} Checking that the coefficients $\psi(\nu,\mu;t)$ are stable under the transition of the underlying lattice from $[1,...,n]$ to $\mathbb{Z}$, and that they match with the duality functions of \cite{bcs}.

\subsection{Occupation and position notation}

Let us first make contact between our notation and that used in \cite{bcs}. The ASEP generator in \cite{bcs} makes particles jump to the left at rate ${\sf p}$ and to the right at rate ${\sf q}$, and is expressed in terms of occupation data $\{\eta_i\}_{i \in \mathbb{Z}}$, where $\eta_i \in \{0,1\}$. In our setting, ${\sf p} = 1$ and ${\sf q} = t$, and the generator is also expressed in terms of occupation data $\{\nu_i\}_{i \in \mathbb{Z}}$.\footnote{A set of inhomogeneous rate parameters $\{a_i\}_{i \in \mathbb{Z}}$ are also employed in \cite{bcs}; we take all such parameters to be $1$.} Summing \eqref{eq:Lpsi} over all $i \in \mathbb{Z}$ and manipulating the summand slightly, we see that
\begin{align}
\label{Lbcs}
\sum_{i \in \mathbb{Z}}
L_i \left[ \psi(\cdot,\mu)\right](\nu) = \sum_{i \in \mathbb{Z}} \Big( t \nu_i(1-\nu_{i+1}) + (1-\nu_i)\nu_{i+1}\Big)
\Big[\psi(s_i \nu, \mu) - \psi(\nu,\mu) \Big],
\end{align}
which matches $L^{\rm occ}$ in \cite{bcs} under the identifications listed above. The reversed ASEP generator in \cite{bcs} makes particles jump to the left at rate ${\sf q}$ and to the right at rate ${\sf p}$, and is expressed in terms of position data $\vec{x} = \{x_i\}_{1 \leq i \leq m}$, where $x_i \in \mathbb{Z}$ is the position of the $i$-th particle. By abuse of notation, we let $\psi(\nu,\mu) \equiv \psi(\nu,\vec{x})$, where we have translated from occupation to position notation in the second argument of $\psi$. Summing \eqref{eq:Mpsi} over all $i \in \mathbb{Z}$ and converting to the position notation, we find that
\begin{align}
\label{Mbcs}
\sum_{i \in \mathbb{Z}}
M_i \left[\psi(\nu,\cdot)\right](\vec{x}) =
\sum_{k\in \ell(\vec{x})}  t \Big( \psi(\nu, \vec{x}_k^{-}) - \psi(\nu,\vec{x}) \Big)
+
\sum_{k\in r(\vec{x})} \Big( \psi(\nu, \vec{x}_k^{+}) - \psi(\nu,\vec{x}) \Big),
\end{align}
where $\ell(\vec{x})$ and $r(\vec{x})$ denote the positions of the leftmost and rightmost particles across all particle clusters, and where ${\vec{x}_k}^{\pm} := (x_1,\dots,x_{k-1},x_k\pm1,x_{k+1},\dots,x_m)$. This matches the reversed generator $L^{\rm part}$ in \cite{bcs}.

\begin{thm}[Sch\"utz \cite{schuetz97}, Borodin--Corwin--Sasamoto \cite{bcs}]
\label{thm:bcs}
Let $\nu$ be an infinite composition with parts $\nu_i \in \{0,1\}$ and fix an ordered $m$-tuple of integers $\vec{x}(\mu) = (x_1 < \cdots < x_m)$, which label the positions of ones in another composition $\mu$. The functions
\begin{align}
\label{eq:sbcs}
\psi\left(\nu,\mu \right)
=
\prod_{x \in \vec{x}(\mu)}
\left(
\prod_{i < x}
t^{\nu_i}
\right)
\nu_{x}
\end{align}
are well defined, since $\nu_i = 0$ for sufficiently small $i$, and satisfy the local duality relation
\begin{align}
\label{local-rels}
L_i \left[\psi\left(\cdot,\mu \right) \right](\nu) 
= 
M_i\left[\psi(\nu,\cdot)\right]\left(\mu \right),
\quad
\forall\ i \in \mathbb{Z},
\end{align}
where $L_i$ and $M_i$ are given by \eqref{eq:Lpsi} and \eqref{eq:Mpsi}, respectively.
\end{thm}
The rest of this section is devoted to proving Theorem \ref{thm:bcs} within the framework developed in this paper.

\subsection{Reduction from rank-$r$ to rank-one}

\begin{defn}
Let $\mu = (\mu_1,\dots,\mu_n)$ be a composition and $\rho(\mu)$ be given by \eqref{rhomu}. The $m$-staircase of $\mu$, denoted $S_m(\mu)$, is an $n$-component vector defined as follows:
\begin{align*}
S_{m}(\mu)
:=m\mu-\rho(\mu)=
(m\mu_1,\dots,m\mu_n) + w_{\mu}\cdot(1,2,\dots,n),
\end{align*}
where we recall that $w_{\mu} \in S_n$ is the minimal-length permutation such that $\mu = w_\mu \cdot \mu^{+}$.
\end{defn}

\begin{prop}\label{eigenvalue and m staircase}
Let $E_{\mu}$ and $E_{\nu}$ be any two non-symmetric Macdonald polynomials, and let $y_i(\mu;q,t)$ and $y_i(\nu;q,t)$ be their eigenvalues under the action of the Cherednik--Dunkl operator $Y_i$, respectively. Then
\begin{align*}
y_i(\mu;t^{-m},t) = y_i(\nu;t^{-m},t),
\quad
\forall\ 1 \leq i \leq n
\iff
S_m(\mu) = S_m(\nu).
\end{align*}
\end{prop}

\begin{proof}
The eigenvalues $y_i(\mu;q,t)$ and $y_i(\nu;q,t)$ match for all $1 \leq i \leq n$ if and only if \eqref{eq:munu} holds. Setting $q=t^{-m}$ in \eqref{eq:munu} and equating the exponents, it is equivalent to the relation
\begin{align*}
S_m(\mu) = m\mu-\rho(\mu)=m\nu-\rho(\nu) = S_m(\nu).
\end{align*}
\end{proof}

\begin{remark}
Notice that we can also write a weaker version of Proposition \ref{eigenvalue and m staircase},
\begin{align*}
y_i(\mu;t^{-m},t) = y_i(\nu;t^{-m},t)
\quad
\forall\ 1 \leq i \leq n
\implies
S_m(\mu^{+}) \sim S_m(\nu^{+})
\end{align*}
where the equivalence relation $\sim$ is defined as follows:
\begin{align*}
S_m(\mu) \sim S_m(\nu)
\iff\
\exists\ \sigma\ \text{such that}\
S_m(\mu) = \sigma \cdot S_m(\nu).
\end{align*}
In other words, the matching of all eigenvalues is only possible if $S_m(\mu^{+})$ and $S_m(\nu^{+})$ are permutable to each other. This is sometimes more useful that Proposition \ref{eigenvalue and m staircase} itself, since the $m$-staircase of a partition is just given by
\begin{align*}
S_m(\mu^{+})
=
(m\mu^{+}_1,\dots,m\mu^{+}_n)
+
(1,2,\dots,n),
\end{align*}
obviating the need to calculate $\rho(\mu)$. 
\end{remark}

\begin{thm}\label{special rank r}
Let $r$ and $m$ be two positive integers such that $n-rm \geq 0$. Consider the anti-partition $\delta = (0^{n-m},r^m)$, and let $f_{\delta}(z_1,\dots,z_n;q,t)$ be the associated ASEP polynomial. Then ${\rm Coeff}_1[f_{\delta},m] \equiv {\rm Coeff}[f_{\delta},m]$ exists, and we have
\begin{align}
\label{rank-r-red}
{\rm Coeff}[f_{\mu},m]
=
\sum_{\nu \in \sigma(\epsilon)}
\psi(\nu,\mu;t)
z^{\nu},
\quad
\forall\ \mu \in \sigma(\delta),
\end{align}
for appropriate coefficients $\psi(\nu,\mu;t)\equiv \psi(\nu,\mu)$, where $\epsilon = (0^{n-rm},1^{rm})$.
\end{thm}

\begin{proof}
We begin by showing that ${\rm Coeff}[f_{\delta},m]$ exists. To establish this, we need to show that the expression $\frac{1}{1-qt^m}$ appears at most linearly in $f_{\delta}$. Using the summation formula \eqref{sum-rankr} together with the results of Proposition \ref{prop:coeff}, we see that the coefficient $C_1(\mu[2],\mu[1];q,t)$ is the only possible source of the factor $\frac{1}{1-qt^m}$ (indeed, another coefficient $C_j(\mu[j+1],\mu[j];q,t)$ with $j \geq 2$ would need to produce $\frac{1}{1-q^j t^{jm}}$ in order to contribute to this factor, which can never happen since the product in the denominator of \eqref{c_j} ranges maximally up to $i=m$). The existence of ${\rm Coeff}[f_{\delta},m]$ is then immediate.

Let us now apply the result of Proposition \ref{prop:partial_E}, in the case $\mu = \delta$ and $p=1$. We see that
\begin{align}
\label{eq:?}
{\rm Coeff}[f_{\delta},m]
=
{\rm Coeff}[E_{\delta},m]
=
\lim_{q \rightarrow t^{-m}}
(1-qt^{m})
\left(
\sum_{\nu \in \mathcal{E}_{\delta}}
c_{\nu}(q,t) E_{\nu}(z;q,t)
\right)
\end{align}
for some family of coefficients $c_{\nu}(q,t)$ and where the sum is over compositions in the set
\begin{align}
\label{eq:5.5}
\mathcal{E}_{\delta}
=
\left\{\nu : \nu \prec \delta,\ y_{\nu}(w) = y_{\delta}(w)\ \text{at}\ q=t^{-m} \right\}.
\end{align}
We will show that the only possible compositions $\nu$ in the set \eqref{eq:5.5} are rank-one. By Proposition \ref{eigenvalue and m staircase} and the remark immediately following it, all compositions in the set \eqref{eq:5.5} would need to satisfy the $m$-staircase relation
\begin{align}
\label{eq:5.7}
S_m(\delta^{+}) \sim S_m(\nu^{+}),
\quad
|\delta| = |\nu|.
\end{align}
Calculating the $m$-staircase of $\delta^{+}$, we find
\begin{align}
\label{stair1}
S_m(\delta^{+})
=
m \cdot (r^m,0^{n-m})+(1,\dots,n)
=
(\underbrace{rm+1,\dots,rm+m}_{m},\underbrace{m+1,\dots,n}_{n-m}),
\end{align}
where we indicate the cardinalities of the two ``blocks'' in $S_m(\delta^{+})$ underneath, for clarity. On the other hand, in view of the fact that $|\nu| = rm$, the composition $\nu$ must have at least $n-rm$ zeros. We can therefore write the $m$-staircase of its dominant reordering as
\begin{align}
\label{stair2}
S_m(\nu^{+})
=
m \cdot (\nu^{+}_1,\dots,\nu^{+}_{rm},0^{n-rm})+(1,\dots,n)
=
(\underbrace{m \nu^{+}_1 +1,\dots, m \nu^{+}_{rm} + rm}_{rm},\underbrace{rm+1,\dots,n}_{n-rm}).
\end{align}
Comparing the final $n-rm$ parts of the two staircases \eqref{stair1} and \eqref{stair2}, we find that they already agree, without the need to permute their order in any way. Suppressing these parts from both \eqref{stair1} and \eqref{stair2}, the remaining entries of $S_m(\delta^{+})$ are permutable to a ``true'' staircase (with step-size one). Our problem thus simplifies to finding partitions $\lambda$ such that
\begin{align*}
(m+1,\dots,rm + m)
\sim
(m \lambda_1 + 1,\dots,m \lambda_{rm} + rm),
\end{align*}
or, after subtracting $m$ from every component,
\begin{align}
\label{eq:5.6}
(1,\dots,rm)
\sim
(m (\lambda_1-1) + 1,\dots,m(\lambda_{rm}-1) + rm).
\end{align} 
A partition solution $\lambda$ of \eqref{eq:5.6} would need to contain two parts $0 \le \lambda_i,\lambda_j \le r$ such that
\begin{align}
\label{lambda_i}
m (\lambda_i - 1) + i &= 1,
\\
\label{lambda_j}
m (\lambda_j - 1) +j &= rm,
\end{align}
with $1 \le i,j \le rm$. Let us examine the possible resolutions of \eqref{lambda_i}, \eqref{lambda_j}.

\smallskip
 {\bf (a)} If the two parts are equal 
($\lambda_i = \lambda_j$), subtracting \eqref{lambda_i} from \eqref{lambda_j} we find that $j-i = rm-1$, which implies $j=rm$ and $i=1$. This identifies $\lambda_i$ and $\lambda_j$ as the first and last parts of the partition; all intermediate parts are then forced to assume the same value. All freedom is exhausted, and we find $\lambda = (1^{rm})$ as the unique solution in the case $\lambda_i = \lambda_j$.

\smallskip
{\bf (b)} Assume a solution exists with $\lambda_i > \lambda_j$. In that case, subtracting \eqref{lambda_i} from \eqref{lambda_j} leads to the inequality $rm-1 < j - i $. There are no values of $i$ and $j$ for which this holds.

\smallskip
{\bf (c)} Finally, assume a solution exists with $\lambda_i < \lambda_j$. Since $\lambda$ is a partition, this would imply $i > j$. Subtracting \eqref{lambda_i} from \eqref{lambda_j}, we observe the equation $m(\lambda_j-\lambda_i) = rm -1 + i -j$. The value of 
$i-j$ is positive, while $\lambda_j - \lambda_i$ is bounded by $r$ (the parts of $\lambda$ cannot exceed $r$), so the only possible resolution in this case is $\lambda_j = r$, $\lambda_i = 0$, $i-j = 1$. This constrains $\lambda_k = r$ for all $k \leq j$ and $\lambda_k = 0$ for all $k \geq j+1$, and since the total weight of $\lambda$ is $rm$, we find that necessarily $j=m$. We recover the solution 
$\lambda = (r^m,0^{rm-m})$.

\smallskip
Translating these findings to our original setting, we have shown that \eqref{eq:5.7} admits only two types of solutions: compositions 
$\nu$ such that $\nu^{+} = (1^{rm},0^{n-rm})$, or $\nu^{+} = (r^m,0^{n-m})$. The latter solution is tautological, since it lives in the same sector as $\delta$; it follows that the set \eqref{eq:5.5} consists only of rank-one compositions.\footnote{One can easily check that the composition $\nu = (1^{rm-m},0^{n-rm},1^m)$ is a particular solution of the equation $S_m(\delta) = S_m(\nu)$, and in fact the minimal one. However for our purposes the precise ordering of parts in $\nu$ is not of interest, since just a statement about the sector of $\nu$ is good enough.} Rank-one non-symmetric Macdonald polynomials are multilinear in $(z_1,\dots,z_n)$, so the right hand side of \eqref{eq:?} must also have a multilinear dependence. It follows, by the action of inverse Hecke generators on \eqref{eq:?}, that a general polynomial $f_{\mu}$ with $\mu \in \sigma(\delta)$ admits the expansion
\begin{align*}
{\rm Coeff}[f_{\mu},m]
=
\sum_{\nu \in \sigma(\epsilon)}
\psi(\nu,\mu)
z^{\nu},
\quad
\epsilon = (0^{n-rm},1^{rm}).
\end{align*}

\end{proof}

\begin{thm}
\label{special rank r duality}
The coefficients in equation \eqref{rank-r-red} are given by
\begin{align}
\label{rank-1-coeffs}
\psi(\nu,\mu)
=
d(t)\cdot
t^{\Omega(\mu,\nu)}
\cdot
I(\mu,\nu),
\quad
\Omega(\mu,\nu)
=
\sum_{1\leq i<j \leq n}
(\bm{1}_{\mu_i<\mu_j})
(\bm{1}_{\nu_i=\nu_j=1}),
\end{align}
where $I(\mu,\nu)$ denotes the indicator function
\begin{align}
\label{ind-def}
I(\mu,\nu)
=
\left\{
\begin{array}{ll}
0, & \exists\ k : (\mu_k,\nu_k) = (r,0),
\\ \\
1, & \text{otherwise}, 
\end{array}
\right.
\end{align}
and $d(t)$ is an overall common factor of the coefficients, and need not be specified explicitly.
\end{thm}

\begin{proof}
We begin by considering the case $\mu = \delta^{+}$ of \eqref{rank-r-red}, namely, the situation when $\mu$ is the unique partition in the sector $\sigma(\delta)$. Using the matrix product formula \eqref{rankr_mpa} we know that $f_{\delta^{+}}(z_1,\dots,z_n;q,t)$ contains the common factor $\prod_{i=1}^{m} z_i$ (each $A_r(z)$ operator in \eqref{rankr_mpa} has a common factor of $z$), while being a homogeneous polynomial in $(z_1,\dots,z_n)$ of total degree $rm$. In addition, this polynomial is symmetric in the subset of variables $(z_{m+1},\dots,z_n)$. On the other hand, \eqref{rank-r-red} says that ${\rm Coeff}[f_{\delta^{+}},m]$ admits an expansion on the space of multilinear polynomials in $(z_1,\dots,z_n)$; the only possible expansion which respects all of these requirements is
\begin{align*}
{\rm Coeff}[f_{\delta^+},m]
=
d(t) \cdot \prod_{i=1}^{m} z_i \cdot e_{(rm-m)}(z_{m+1},\dots,z_n)
=
d(t) 
\times
\sum_{\nu \in \sigma(\epsilon)}
I(\delta^{+},\nu)
f_{\nu}
\end{align*}
for some constant $d(t)$, where $\epsilon = (0^{n-rm},1^{rm})$ and $I(\delta^{+},\nu)$ is given by \eqref{ind-def}. This confirms the formula \eqref{rank-1-coeffs} for the case $\mu = \delta^{+}$, since one clearly has $\Omega(\delta^{+},\nu) = 0$ for all $\nu$. 

We use the preceding special case as the basis for induction. Let us suppose that $\psi(\nu,\mu)$ is given by \eqref{rank-1-coeffs} for all $\nu$, where $\mu$ is some composition in the sector $\sigma(\delta)$, which contains (at least) one pair of parts ($\mu_i,\mu_{i+1})$ such that $\mu_i > \mu_{i+1}$. We then act on \eqref{rank-r-red} with $T_i$, giving
\begin{align*}
T_i \cdot {\rm Coeff}[f_{\mu},m]
=
{\rm Coeff}[f_{s_i \mu},m]
=
\sum_{\nu \in \sigma(\epsilon)}
\psi(\nu,\mu)
(T_i \cdot f_{\nu}).
\end{align*}
For the action of $T_i$ on $f_{\nu}$, we should distinguish the three possibilities (i) $\nu_i > \nu_{i+1}$, (ii) $\nu_i = \nu_{i+1}$ and (iii) $\nu_i < \nu_{i+1}$, as given by equations \eqref{localtKZ} and \eqref{other_case}. Case (i) means that $(\mu_i,\mu_{i+1}) = (r,0)$ and $(\nu_i,\nu_{i+1}) = (1,0)$, and one easily sees that
\begin{align}
\label{i}
\psi(\nu,\mu)
(T_i \cdot f_{\nu})
=
\psi(\nu,\mu)
f_{s_i \nu}
=
\psi(s_i\nu,s_i\mu)
f_{s_i \nu}.
\end{align}
Case (ii) means that $(\mu_i,\mu_{i+1}) = (r,0)$ and $(\nu_i,\nu_{i+1}) = (1,1)$ (we exclude the possibility that $(\nu_i,\nu_{i+1}) = (0,0)$, since we would then have $(\mu_i,\nu_i) = (r,0)$, causing the indicator function \eqref{ind-def} to vanish), and accordingly,
\begin{align}
\label{ii}
\psi(\nu,\mu)
(T_i \cdot f_{\nu})
=
t
\psi(\nu,\mu)
f_{s_i \nu}
=
\psi(s_i\nu,s_i\mu)
f_{s_i \nu},
\end{align}
where the final equality exploits the fact that in this case $\Omega(\mu,\nu)+1 = \Omega(s_i\mu,s_i\nu)$. Finally, case (iii) means that $(\mu_i,\mu_{i+1}) = (r,0)$ and $(\nu_i,\nu_{i+1}) = (0,1)$, which is another situation where the indicator function \eqref{ind-def} vanishes. We thus have the trivial fact
\begin{align}
\label{iii}
\psi(\nu,\mu)
(T_i \cdot f_{\nu})
=
0
=
\psi(s_i\nu,s_i\mu)
f_{s_i \nu}.
\end{align}
One finds the same expression for the right hand side in all three cases \eqref{i}--\eqref{iii}; we have thus demonstrated that
\begin{align*}
{\rm Coeff}[f_{s_i \mu},m]
=
\sum_{\nu \in \sigma(\epsilon)}
\psi(s_i\nu,s_i\mu)
f_{s_i \nu}
=
\sum_{\nu \in \sigma(\epsilon)}
\psi(\nu,s_i\mu)
f_{\nu},
\end{align*}
which is the required inductive step. This completes the proof of \eqref{rank-1-coeffs}.
\end{proof}

\subsection{Back to the proof of Theorem \ref{thm:bcs}}
\label{ssec:back}

In the previous subsection we started from a rank-$r$ ASEP polynomial $f_{\mu}$ such that 
$\mu^{-} = (0^{n-m},r^{m})$, and sent $q \rightarrow t^{-m}$. Quite remarkably, one finds that ${\rm Coeff}[f_{\mu},m]$ reduces to a linear combination of rank-one ASEP polynomials $f_{\nu}$ such that $\nu^{-} = (0^{n-rm},1^{rm})$, where the expansion coefficients are given by \eqref{rank-1-coeffs}. Applying the result of Theorem \ref{thm:central}, we now obtain the desired duality statement:
\begin{cor}
In the same notation as Theorem \ref{special rank r duality}, the functions
\begin{align}
\label{eq:dual_1}
\psi(\nu,\mu) = t^{\Omega(\mu,\nu)} \cdot I(\mu,\nu)
\end{align} 
satisfy the local duality relations
\begin{align}
\label{eq:dual_2}
L_i [\psi(\cdot,\mu)](\nu) = M_i [\psi(\nu,\cdot)](\mu), \qquad\forall\ 1 \leq i \leq n-1,
\end{align}
where the left hand side is given by \eqref{eq:Lpsi}, and the right hand side by \eqref{eq:Mpsi}. Note that we have dropped the constant $d(t)$ from \eqref{eq:dual_1}; we are allowed to do this because it is common to all coefficients $\psi(\nu,\mu)$ in the sectors we have chosen, and therefore plays no role in \eqref{eq:dual_2}.
\end{cor}

\begin{remark}
Even though we used a higher-rank ASEP polynomial $f_{\mu}$ in the derivation of this duality statement, it is clear that  \eqref{eq:dual_2} itself is a rank-one equation: because of the sector that $\mu$ belongs to, the $L_i$ generator sees only particles of type $r$ and zeros, and so the left hand side of \eqref{eq:dual_2} describes the evolution of an ordinary (single-species) ASEP.
\end{remark}
To complete the proof of Theorem \ref{thm:bcs}, one should translate the observable \eqref{eq:dual_1} into the occupation--position notation employed therein. With $\vec{x}(\mu) = (x_1(\mu) < \cdots < x_m(\mu))$ denoting the positions of the $r$-particles in the composition $\mu$, after a simple calculation one finds that
\begin{align*}
\psi\left(\nu,\mu\right)
=
t^{-m(m-1)/2}
\times
\prod_{j=1}^{m}
\left(
\prod_{1 \leq i < x_j(\mu)}
t^{\nu_i}
\right)
\nu_{x_j(\mu)},
\end{align*}
which matches the form of the right hand side of \eqref{eq:sbcs} up to the factor $t^{-m(m-1)/2}$. This factor is spurious; it does not play any role in the equations \eqref{eq:dual_2} other than as a spectating constant.

Finally, our analysis so far has proceeded on the finite lattice $[1,\dots,n]$. It is a trivial matter to transition to the integer lattice. Indeed, the observable \eqref{eq:dual_1} does not depend on $n$ in any way (beyond the fact that it is the length of the participating compositions). One can therefore embed the existing observables within the space of functions on $\mathbb{Z} \times \mathbb{Z}$, simply by padding the finite compositions $\mu$ and $\nu$ with zeros on both sides. This reproduces the family of observables \eqref{eq:sbcs}, and finishes our derivation of Theorem \ref{thm:bcs}.

\section{Rank-two ASEP dualities}
\label{se:rank2}

The aim of this section is to produce new types of observables, which generalize those found in \cite{bcs}, being duality functions with respect to two {\it multi-species} asymmetric simple exclusion processes. We will restrict our attention to dualities between mASEPs with {\it two} distinct particle species, in this way finding a natural rank-two extension of Theorem \ref{thm:bcs}. 

For other recent progress related to higher-rank duality functions, making use of quantum group symmetries, we refer the reader to \cite{carinciGRSa,carinciGRSb,kuan,kuan2}.

\subsection{Reduction relations between a pair of rank-two sectors}
\label{sec:rank2-red}

\begin{thm}\label{thm:rank2-sector}
Fix three integers $n,m_1,m_2 \geq 0$ such that $m_1+m_2 \leq n$, and an anti-partition
\begin{align*}
\delta = (0^{n-m_1-m_2},1^{m_1},2^{m_2}),
\end{align*}
Then choosing another integer $p$ such that $1 \leq p \leq \min(n-m_1-m_2,m_2)$, one has the expansion
\begin{align}
\label{rank_2_red}
{\rm Coeff}[f_{\mu},p+m_1]
=
\sum_{\nu \in \sigma(\epsilon)}
\psi(\nu,\mu;t)
f_{\nu}(z;t^{-p-m_1},t),
\quad
\forall\ \mu \in \sigma(\delta),
\end{align}
for appropriate coefficients $\psi(\nu,\mu;t)\equiv \psi(\nu,\mu)$, where
\begin{align*}
\epsilon = (0^{n-m_1-m_2-p},1^{m_1+2p},2^{m_2-p}).
\end{align*}
\end{thm}

\begin{proof}
Let us begin by remarking that this theorem is not obvious from the matrix product formula \eqref{MPA}, for although the latter allows us to manually calculate ${\rm Coeff}[f_{\mu},p+m_1]$, the resulting expression is not easily re-expressed in the basis of the polynomials $f_{\nu}$. 

It is therefore best to resort to a similar style of proof as that of Theorem \ref{special rank r}. In the present situation, given that our starting sector (the sector of $\delta$) is rank-two, rather than rank-$r$, we are able to be a little more explicit. We will show that all members of the set
\begin{align}
\label{eq:6.1}
\mathcal{E}_{\delta}
=
\left\{\nu : \nu \prec \delta,\ y_{\nu}(w) = y_{\delta}(w)\ \text{at}\ q=t^{-p-m_1} \right\}
\end{align}
live in the composition sector $\sigma(\epsilon)$, allowing us to conclude that
\begin{align}
\label{eq:6.6}
{\rm Coeff}[f_{\delta},p+m_1]
=
{\rm Coeff}[E_{\delta},p+m_1]
=
\lim_{q\rightarrow t^{-p-m_1}}
(1-qt^{p+m_1})
\left(
\sum_{\nu \in \sigma(\epsilon)}
c_{\nu}(q,t) E_{\nu}(z;q,t)
\right),
\end{align}
for some family of coefficients $c_{\nu}(q,t)$. Any compositions in \eqref{eq:6.1} would need to have the same weight as $\delta$, with parts of at most size two, so it is clearly sufficient to restrict our search to compositions that have the dominant ordering
\begin{align*}
\nu^{+}
=
(2^{m_2-r},1^{m_1+2r},0^{n-m_1-m_2-r}),
\end{align*}
with $r \geq 1$ becoming the only degree of freedom. Our aim is to prove that $r=p$ is the only possible value for $r$, which we do by exhausting all solutions of the relation $S_{p+m_1}(\delta^{+}) \sim S_{p+m_1}(\nu^{+})$. With $\nu^+$ as above and $\delta^+=(2^{m_2},1^{m_1},0^{n-m_1-m_2})$ we see that 
\begin{align*}
\delta^+_i=\nu^+_i,
\quad
\forall\
i\in[1,m_2-r]\cup[m_2+1,m_1+m_2]\cup[m_1+m_2+r+1,n],
\end{align*}
accordingly $S_{p+m_1}(\delta^+)_i=S_{p+m_1}(\nu^+)_i$ for these values of $i$. Thus it suffices to study instead the relation
\begin{align}
\label{eq:6.2}
\S_1(\delta^+)\cup\S_2(\delta^+) \sim \S_1(\nu^+)\cup\S_2(\nu^+),
\end{align}
where
\begin{align*}
\S_1(\mu) &=\{S_{p+m_1}(\mu)_i|\ i\in \mathcal{A}_1\},\qquad \S_2(\mu)=\{S_{p+m_1}(\mu)_i|\ i\in \mathcal{A}_2\},
\\
\mathcal{A}_1 &=[m_2-r+1,m_2], \qquad \mathcal{A}_2=[m_1+m_2+1,m_1+m_2+r].
\end{align*}
Let us first suppose that $r>p$. Consider the following component of $\mathcal{S}_{1}(\nu^{+})$, corresponding with the lowest index in $\mathcal{A}_1$:
\begin{align}
\label{eq:Snu}
S_{p+m_1}(\nu^+)_{m_2-r+1}
=
(p+m_1) \cdot \nu^+_{m_2-r+1}+m_2-r+1
=
m_1+m_2+1+p-r.
\end{align}
This element must be reproduced somewhere in $\S_1(\delta^+)\cup\S_2(\delta^+)$, or the relation \eqref{eq:6.2} does not hold.
It is easy to check that the smallest element in $\S_1(\delta^+)$ is given by
\begin{align*}
S_{p+m_1}(\delta^+)_{m_2-r+1}
=
2m_1+2p+m_2-r+1.
\end{align*}
Clearly $S_{p+m_1}(\delta^+)_{m_2-r+1} >S_{p+m_1}(\nu^+)_{m_2-r+1}$ and hence there is no element in the set $\S_1(\delta^+)$ which reproduces the value on the right hand side of \eqref{eq:Snu}. Similarly, the smallest element in $\S_2(\delta^+)$ is given by
\begin{align}
\label{eq:6.4}
S_{p+m_1}(\delta^+)_{m_1+m_2+1}
=
m_1+m_2+1,
\end{align}
and since by assumption $r>p$, it follows that $S_{p+m_1}(\nu^+)_{m_2-r+1} < S_{p+m_1}(\delta^+)_{m_1+m_2+1}$. We conclude that there is also no element in $\S_2(\delta^+)$ with value matching the right hand side of \eqref{eq:Snu}. Thus for $r>p$, the relation \eqref{eq:6.2} has no solutions.

Second, we suppose that $r<p$. Consider the component of $\S_2(\delta^{+})$ corresponding with the lowest index in $\mathcal{A}_2$, as given by \eqref{eq:6.4}. This element must be reproduced somewhere in $\S_1(\nu^+)\cup\S_2(\nu^+)$. Since $\nu^+_i=1$ for all $i\in \mathcal{A}_1\cup \mathcal{A}_2$, the smallest element in $\S_1(\nu^+)\cup\S_2(\nu^+)$ is obtained by taking the first index in $\mathcal{A}_1$. We then find that 
\begin{align*}
S_{p+m_1}(\nu^+)_{m_2-r+1}
=
m_1+m_2+1+p-r
>
m_1+m_2+1
=
S_{p+m_1}(\delta^+)_{m_1+m_2+1}.
\end{align*}
Hence there is no element in $\S_1(\nu^+)\cup\S_2(\nu^+)$ which reproduces the right hand side of \eqref{eq:6.4}, and accordingly the relation \eqref{eq:6.2} has no solutions for $r < p$. 

We have shown that compositions $\nu$ such that $\nu^{+} = (2^{m_2-p},1^{m_1+2p},0^{n-m_1-m_2-p})$ are the only possible members of the set \eqref{eq:6.1}. From here it is quite straightforward to see that
\begin{align*}
\nu = (1^p,0^{n-m_1-m_2-p},2^{m_2-p},1^{m_1+p})
\end{align*}
satisfies $S_{p+m_1}(\delta) = S_{p+m_1}(\nu)$, and is the minimal such composition. The claim \eqref{eq:6.6} is proved; one can now follow a similar procedure as in the proof of Theorem \ref{thm:sector}, to transform the right hand side of \eqref{eq:6.6} to the basis of ASEP polynomials. This leads to the generic expansion \eqref{rank_2_red}.
\end{proof}

\begin{thm}
\label{thm:generic}
The coefficients in equation \eqref{rank_2_red} are given by
\begin{align}
\label{coeff_2}
\psi(\nu,\mu)
=
d(t)\cdot
t^{\Omega(\mu,\nu)}
\cdot
I(\mu,\nu),
\quad
\Omega(\mu,\nu)
=
\sum_{1\leq i<j \leq n}
(\bm{1}_{\mu_i<\mu_j})
(\bm{1}_{\nu_i=\nu_j=1}),
\end{align}
where $I(\mu,\nu)$ denotes the indicator function
\begin{align}
\label{ind-def2}
I(\mu,\nu)
=
\left\{
\begin{array}{ll}
0, & \exists\ k : \mu_k > \nu_k =0,\ \  \text{or}\ \ \mu_k < \nu_k =2,
\\ \\
1, & \text{otherwise}, 
\end{array}
\right.
\end{align}
and $d(t)$ is an overall common factor of the coefficients.
\end{thm}

\begin{proof}
We start from the generic expansion \eqref{rank_2_red}, as given to us by Theorem \ref{thm:rank2-sector}. Since $z^{\nu}$ is the leading monomial of the monic polynomial $f_{\nu}(z;t^{-p-m_1},t)$, and unique to that polynomial on the right hand side of \eqref{rank_2_red}, we can evaluate $\psi(\nu,\mu)$ by taking the coefficient of $z^{\nu}$ in ${\rm Coeff}[f_{\mu},p+m_1]$. We then use the matrix product formula \eqref{MPA} to perform the calculation:
\begin{align}
\label{coeff_3}
\psi(\nu,\mu)
=
\left[
\lim_{q \rightarrow t^{-p-m_1}} \left(1-q t^{p+m_1}\right)
{\rm Tr}
\Big(
A_{\mu_1}(z_1)
\dots
A_{\mu_n}(z_n)
k^{u}
\Big)
\right]_{z^{\nu}}
\end{align}
and noting the $z$-dependence of the operators $A_i(z)$ in \eqref{ZF-ops}, we immediately see that $\psi(\nu,\mu)$ is zero if for some 
$1 \leq k \leq n$ we have $\mu_k > \nu_k = 0$ or $\mu_k < \nu_k = 2$. This is the reason why the coefficients \eqref{coeff_2} contain the indicator function \eqref{ind-def2}; we restrict our attention henceforth to the situation when $\nu$ is chosen such that $I(\mu,\nu)$ is non-zero. Using \eqref{coeff_3}, we see that
\begin{align}
\label{constant-MPA}
\psi(\nu,\mu)
=
\lim_{q \rightarrow t^{-p-m_1}} \left(1-q t^{p+m_1}\right)
{\rm Tr}
\Big(
B_{\mu_1,\nu_1} \dots B_{\mu_n, \nu_n} k^u
\Big)
\cdot
I(\mu,\nu),
\end{align}
with $B_{0,0} = B_{2,2} = 1$, $B_{1,1} = k$, $B_{0,1} =\phi$ and $B_{2,1} = \phid$. Since the part-multiplicities of $\nu$ are already specified by Theorem \ref{thm:rank2-sector}, we can assume that $\#\{B_{0,1}\} = \#\{B_{2,1}\} = p$ and $\#\{B_{1,1}\} = m_1$. The product of bosonic operators appearing in \eqref{constant-MPA} can then be brought, via repeated use of the relations $\phi \phid = 1-tk$ and $\phid \phi = 1-k$, to a polynomial in $k$:
\begin{align}
\label{bosonic-poly}
B_{\mu_1,\nu_1} \dots B_{\mu_n, \nu_n}
=
\sum_{i=0}^{p}
c_{\mu,\nu}(i;t)
k^{i+m_1},
\end{align}
for suitable coefficients $c_{\mu,\nu}(i;t)$, which for the moment we do not specify. Substituting this into \eqref{constant-MPA} and evaluating the resulting traces, we find
\begin{align*}
\psi(\nu,\mu)
=
\lim_{q \rightarrow t^{-p-m_1}} \left(1-q t^{p+m_1}\right)
\left(
\sum_{i=0}^{p}
\frac{c_{\mu,\nu}(i;t)}{1-qt^{i+m_1}}
\right)
\cdot
I(\mu,\nu)
=
c_{\mu,\nu}(p;t)
\cdot
I(\mu,\nu).
\end{align*}
It is straightforward to calculate the top-degree term in \eqref{bosonic-poly}. In the case of a completely ordered string of bosonic operators, one has
\begin{align*}
\underbrace{\phid \dots \phid}_p
\underbrace{k \dots k}_{m_1}
\underbrace{\phi \dots \phi}_p
=
d_{\mu}(p;t)
k^{p+m_1}
+
\text{subleading terms in}\ k,
\end{align*}
where $d_{\mu}(p;t) = (-1)^p (t^{-p})^{m_1+(p-1)/2}$. As the string becomes disordered, one easily sees that the leading coefficient acquires a factor of $t$ for every pair $\phi \dots \phid$, $k \dots \phid$ or $\phi \dots k$ that gets created. These pairs are counted by
\begin{align*}
\alpha(\mu,\nu)
&=
\#\{i<j | (\mu_i=0,\mu_j=2), (\nu_i=\nu_j=1) \},
\\
\beta(\mu,\nu)
&=
\#\{i<j | (\mu_i=1,\mu_j=2), (\nu_i=\nu_j=1) \},
\\
\gamma(\mu,\nu)
&=
\#\{i<j | (\mu_i=0,\mu_j=1), (\nu_i=\nu_j=1) \},
\end{align*}
respectively. We conclude that
\begin{align*}
\psi(\nu,\mu)
=
c_{\mu,\nu}(p;t)
\cdot
I(\mu,\nu)
=
d_{\mu}(p;t)
\times
t^{\alpha(\mu,\nu)+\beta(\mu,\nu)+\gamma(\mu,\nu)}
\cdot
I(\mu,\nu),
\end{align*}
completing the proof of \eqref{coeff_2}, with the identification $d(t) \equiv d_{\mu}(p;t)$.
\end{proof}

\subsection{Rank-two duality functions} 

In the last subsection we studied the reduction of a generic rank-two ASEP polynomial $f_{\mu}$, such that 
$\mu^{-} = (0^{n-m_1-m_2},1^{m_1},2^{m_2})$, in the limit $q \rightarrow t^{-p-m_1}$ with $p$ a positive integer. In Theorem \ref{thm:rank2-sector} we proved that the corresponding expansion over polynomials $f_{\nu}$ is contained to the sector in which $\nu^{-} = (0^{n-m_1-m_2-p},1^{m_1+2p},2^{m_2-p})$, and in Theorem \ref{thm:generic} we calculated the expansion coefficients. By virtue of Theorem \ref{thm:central}, we have proved the following duality result:
\begin{cor}
\label{cor:rank2}
In the same notations as Theorem \ref{thm:generic}, the functions
\begin{align}
\label{eq:dual_3}
\psi(\nu,\mu) = t^{\Omega(\mu,\nu)} \cdot I(\mu,\nu)
\end{align} 
satisfy the local duality relations
\begin{align}
\label{eq:dual_4}
L_i [\psi(\cdot,\mu)](\nu) = M_i [\psi(\nu,\cdot)](\mu), \qquad\forall\ 1 \leq i \leq n-1,
\end{align}
where the left hand side is given by \eqref{eq:Lpsi}, and the right hand side by \eqref{eq:Mpsi}.
\end{cor}

Let us now translate the observable \eqref{eq:dual_3} into an occupation--position notation, similar to that employed in Theorem \ref{thm:bcs}. In the rank-two case at hand, the composition $\mu$ is labelled by two sets of positions: a set $\vec{x}(\mu) = (x_1 < \cdots < x_{m_1})$ which labels the positions of 1-particles, and a set $\vec{y}(\mu) = (y_1 < \cdots < y_{m_2})$ labelling the positions of 2-particles. The two sets $\vec{x}$ and $\vec{y}$ are disjoint, since two particles cannot occupy a single site of the lattice. We introduce a statistic $\chi(\vec{x},\vec{y})$, which counts the number of ``crossings'' between the two sets $\vec{x}$ and $\vec{y}$:
\begin{align}
\label{chi}
\chi(\vec{x},\vec{y}) := \#\{(x_i,y_j) \in (\vec{x},\vec{y})\ |\ x_i>y_j\}.
\end{align}
\begin{prop}
\label{prop:exponent}
Fix two compositions 
\begin{align*}
\mu \in \sigma(0^{n-m_1-m_2},1^{m_1},2^{m_2}),
\quad\quad
\nu \in \sigma(0^{n-m_1-m_2-p},1^{m_1+2p},2^{m_2-p}),
\end{align*}
chosen such that the inequalities $\mu_k > \nu_k = 0$ and $\mu_k < \nu_k = 2$ do not occur for any $1 \leq k \leq n$. Let $\Omega(\mu,\nu)$ be given by \eqref{coeff_2}. Expressing $\mu$ in terms of particle-position notation, one has
\begin{align}
\label{exponent}
\Omega(\mu,\nu)
+
\frac{m_1(m_1-1)}{2} + \frac{p(p-1)}{2}
+
\chi(\vec{x},\vec{y})
=
\sum_{x \in \vec{x}(\mu)}
\sum_{i<x}
\bm{1}_{\nu_i \geq 1}
+
\sum_{y \in \vec{y}(\mu)}
\sum_{i<y}
\bm{1}_{\nu_i = 1}
\bm{1}_{\nu_y = 1}.
\end{align}
\end{prop}

\begin{proof}
We start from the left hand side of \eqref{exponent}, and examine what it counts. In the following we always assume that $i<j$.
\begin{itemize}
\item The first term, $\Omega(\mu,\nu)$, counts all instances such that $(\mu_i,\mu_j) = (0,1), (0,2), (1,2)$ and $(\nu_i,\nu_j) = (1,1)$. 

\smallskip

\item The second term, $m_1(m_1-1)/2$, is equal to the number of times that $(\mu_i,\mu_j) = (1,1)$ and $(\nu_i,\nu_j) = (1,1)$ (since $\mu_i = 1$ forces $\nu_i =1$, by our assumption on the compositions). 

\smallskip

\item The third term, $p(p-1)/2$, is equal to the number of times that $(\mu_i,\mu_j) = (2,2)$ and $(\nu_i,\nu_j) = (1,1)$. To see this, note that there must be exactly $p$ pairs $(\mu_i,\nu_i) = (2,1)$, by knowledge of the sectors that the two compositions come from.

\smallskip

\item The fourth term, $\chi(\vec{x},\vec{y})$, counts the number of times that $(\mu_i,\mu_j) = (2,1)$ and $(\nu_i,\nu_j) = (1,1), (2,1)$.
\end{itemize}
Totaling these possibilities, we find that the left hand side counts $7$ different types of pairs $(\mu_i,\mu_j)$, $(\nu_i,\nu_j)$. We proceed to show that the same pairs are recovered on the right hand side of \eqref{exponent}:
\begin{itemize}

\item The first summation, $\sum_{x \in \vec{x}(\mu)} \sum_{i<x} \bm{1}_{\nu_i \geq 1}$, counts all instances such that $\nu_i \geq 1,\mu_j = 1$. This can be seen to be equal to
\begin{align*}
&
\#\Big\{ (\mu_i,\mu_j) = (0,1), (\nu_i,\nu_j) = (1,1) \Big\} 
+
\#\Big\{ (\mu_i,\mu_j) = (1,1), (\nu_i,\nu_j) = (1,1) \Big\}
+
\\
&
\#\Big\{ (\mu_i,\mu_j) = (2,1), (\nu_i,\nu_j) = (1,1) \Big\}
+
\#\Big\{ (\mu_i,\mu_j) = (2,1), (\nu_i,\nu_j) = (2,1) \Big\},
\end{align*}
by virtue of the restrictions imposed on $\mu$ and $\nu$. This accounts for $4$ of the terms on the left hand side of \eqref{exponent}.

\item The second summation, $\sum_{y \in \vec{y}(\mu)} \sum_{i<y} \bm{1}_{\nu_i = 1} \bm{1}_{\nu_y = 1}$, enumerates all the instances such that $\nu_i = \nu_j = 1,\mu_j = 2$. More explicitly, these instances are given by
\begin{align*}
&
\#\Big\{ (\mu_i,\mu_j) = (0,2), (\nu_i,\nu_j) = (1,1) \Big\} 
+
\#\Big\{ (\mu_i,\mu_j) = (1,2), (\nu_i,\nu_j) = (1,1) \Big\}
+
\\
&
\#\Big\{ (\mu_i,\mu_j) = (2,2), (\nu_i,\nu_j) = (1,1) \Big\}.
\end{align*}
This accounts for the remaining $3$ types of terms on the left hand side of \eqref{exponent}.
\end{itemize}
\end{proof}
Using the result of Proposition \ref{prop:exponent} we can now write the observable in Corollary \ref{cor:rank2} as
\begin{align}
\label{rank2-local-dual}
\psi\left(\nu,\mu\right)
=
\prod_{x \in \vec{x}(\mu)}
\prod_{i < x}
\left(
t^{\bm{1}_{\nu_i \geq 1}}
\right)
\cdot
\prod_{y \in \vec{y}(\mu)}
\prod_{i < y}
\left(
t^{\bm{1}_{\nu_i = 1} \bm{1}_{\nu_y=1}}
\right)
\cdot
t^{-\chi(\vec{x},\vec{y})}
\cdot
I(\mu,\nu),
\end{align}
where we have dropped an irrelevant overall factor of $t^{-(m_1(m_1-1)+p(p-1))/2}$, which comes from \eqref{exponent}, but plays no role in the duality relations \eqref{eq:dual_4}. One can view \eqref{rank2-local-dual} as a duality function on $\mathbb{Z} \times \mathbb{Z}$, by extending the compositions $\mu,\nu$ to infinite length in a stable way, as we discussed in Section \ref{ssec:back}.

The duality function \eqref{rank2-local-dual} is a generalization of the observable \eqref{eq:sbcs} to the rank-two setting. It is easily seen to degenerate to the latter when both $\mu$ and $\nu$ are chosen to have no $2$-particles. Another special case of interest is when $\mu$ is a generic rank-two composition, while $\nu$ is purely rank-one. In that case, \eqref{rank2-local-dual} simplifies, and we obtain the following result:
\begin{cor}
Let $\nu$ be an infinite rank-one composition, with $\nu_i \in \{0,1\}$ for all $i \in \mathbb{Z}$. Let $\mu$ be an infinite rank-two composition, with $\vec{x}(\mu) = (x_1 < \cdots < x_{m_1})$ and $\vec{y}(\mu) = (y_1 < \cdots < y_{m_2})$ labelling the positions of its $1$ and $2$-particles, respectively.  Then the function
\begin{align}
\label{rank2:sbcs}
\psi\left(\nu,\mu\right)
=
\prod_{x \in \vec{x}(\mu)}
\prod_{i < x}
\left(
t^{\nu_i}
\right)
\nu_x
\cdot
\prod_{y \in \vec{y}(\mu)}
\prod_{i < y}
\left(
t^{\nu_i}
\right)
\nu_y
\cdot
t^{-\chi(\vec{x},\vec{y})}
\end{align}
satisfies the relations \eqref{eq:dual_4} for all $i \in \mathbb{Z}$.
\end{cor}

\section{Duality functions without indicators}
\label{se:concl}

In Section \ref{se:rank1reduc} we presented a new derivation of the rank-one observable \eqref{eq:sbcs}, proving that it is a solution of the {\it local}\/ duality equations \eqref{local-rels}. A second observable was considered in \cite{bcs}. This observable differs from \eqref{eq:sbcs} in two main ways: first, it does not contain any indicator functions, meaning that the observable does not vanish for any values of $\mu$ and $\nu$; second, the resulting observable does not satisfy the local relations \eqref{local-rels}, but rather the {\it global}\/ relation obtained by summing over all $i \in \mathbb{Z}$. In Section \ref{ssec:7.1} we briefly review these facts. 

In Section \ref{ssec:7.2}, we will show that the rank-two observable obtained in equation \eqref{rank2:sbcs} also gives rise to a ``partner'' observable without indicator functions, which satisfies global duality relations. Once this result is written down, it is not hard to see that it in fact generalizes to arbitrary rank: hence in Section \ref{ssec:7.3} we find a non-vanishing observable valued on a rank-one ASEP and an arbitrary rank mASEP, which is a duality function with respect to the generators of the two processes.

\subsection{Rank-one duality functions without indicators}
\label{ssec:7.1}

\begin{prop}
[Borodin--Corwin--Sasamoto \cite{bcs}]
\label{prop:7.1}
Let $\nu$ be an infinite composition with parts $\nu_i \in \{0,1\}$ and fix an ordered $m$-tuple of integers $\vec{x}(\mu) = (x_1 < \cdots < x_m)$, which label the positions of ones in another composition $\mu$. Then the observable 
\begin{align}
\label{glob-rnk1}
H\left(\nu,\mu\right)
=
\prod_{x \in \vec{x}(\mu)}
\prod_{i \leq x}
t^{\nu_i}
\end{align}
satisfies the equation
\begin{align}
\label{global-bcs}
\sum_{i \in \mathbb{Z}} 
L_i\left[H\left(\cdot,\mu\right)\right](\nu) 
= 
\sum_{i \in \mathbb{Z}}
M_i\left[H(\nu,\cdot)\right]\left(\mu\right),
\end{align}
where $L_i$ and $M_i$ are given by \eqref{eq:Lpsi} and \eqref{eq:Mpsi}, respectively.
\end{prop}

\begin{proof}
Once guessed, this result can be proved by direct computation; see \cite{bcs}. We do not know of a more constructive proof, for example making use of solutions of the $t$KZ equations, although it would be very interesting to find one.
\end{proof}

\subsection{Rank-one/rank-two duality functions without indicators}
\label{ssec:7.2}

\begin{prop}
\label{prop:7.2}
Let $\nu$ be an infinite rank-one composition, with $\nu_i \in \{0,1\}$ for all $i \in \mathbb{Z}$. Let $\mu$ be an infinite rank-two composition, with $\vec{x}(\mu) = (x_1 < \cdots < x_{m_1})$ and $\vec{y}(\mu) = (y_1 < \cdots < y_{m_2})$ labelling the positions of its $1$ and $2$-particles, respectively. Recall also the definition of $\chi(\vec{x},\vec{y})$, as given by \eqref{chi}. Then the observable
\begin{align}
\label{glob-rnk2}
H\left(\nu,\mu\right)
=
\prod_{x \in \vec{x}(\mu)}
\prod_{i \leq x}
\left(
t^{\nu_i}
\right)
\cdot
\prod_{y \in \vec{y}(\mu)}
\prod_{i \leq y}
\left(
t^{\nu_i}
\right)
\cdot
t^{-\chi(\vec{x},\vec{y})}
\end{align}
satisfies the equation
\begin{align}
\label{aim-globaleqn}
\sum_{i \in \mathbb{Z}} 
L_i\left[H\left(\cdot,\mu\right)\right](\nu) 
= 
\sum_{i \in \mathbb{Z}}
M_i\left[H(\nu,\cdot)\right]\left(\mu\right),
\end{align}
where $L_i$ and $M_i$ are given by \eqref{eq:Lpsi} and \eqref{eq:Mpsi}, respectively.
\end{prop}

\begin{proof}
It is convenient to define the set
\begin{align*}
\vec{z}(\mu) = (z_1 < \cdots < z_{m_1+m_2}) = \vec{x}(\mu) \cup \vec{y}(\mu),
\end{align*}
obtained by taking the union of the two sets of coordinates $\vec{x}$ and $\vec{y}$. We begin by considering the case of a single particle cluster in $\mu$. This refers to the situation in which $\vec{z}(\mu)=(z+1,\dots,z+l)$ for some $z \in \mathbb{Z}$, and where we abbreviate $m_1+m_2 \equiv l$. 

To begin, notice that the observable \eqref{glob-rnk2} can be expressed in the form 
\begin{align}
\label{global2-1}
H\left(\nu,\mu\right) 
=
H\left(\nu,\mu^{*}\right) t^{-\chi(\vec{x},\vec{y})},
\end{align}
where $\mu^{*}$ is the rank-one composition obtained by the following ``colour-blind'' projection of $\mu$:
\begin{align*}
\mu^{*}_i 
= 
\left\{
\begin{array}{ll}
0, & \quad \mu_i = 0,
\\
1, & \quad \mu_i \geq 1,
\end{array}
\right.
\quad
\quad
\forall\ i \in \mathbb{Z},
\end{align*}
and $H\left(\nu,\mu^{*}\right)$ denotes a rank-one observable of the form \eqref{glob-rnk1}:
\begin{align}
H\left(\nu,\mu^* \right)
=
\prod_{x \in \vec{x}(\mu^*)}
\prod_{i \leq x}
t^{\nu_i}.
\end{align}
Studying firstly the right hand side of the proposed identity \eqref{aim-globaleqn}, we see from the action \eqref{eq:Mpsi} of $M_i$ that we can localise the summation over $i$ as follows:
\begin{align}
\label{smaller-rhs}
\sum_{i \in \mathbb{Z}}
M_i [H(\nu,\cdot)]\left(\mu \right)
=
\sum_{i\in\{z,z+l\}\bigcup d_1(\vec{x},\vec{y})\bigcup d_2(\vec{x},\vec{y})}
M_i[H(\nu,\cdot)](\mu),
\end{align}
where we have defined the sets
\begin{align*}
d_1(\vec{x},\vec{y})=\{x_i\in\vec{x}|x_i+1\in\vec{y}\},
\quad\quad
d_2(\vec{x},\vec{y})=\{y_i\in\vec{y}|y_i+1\in\vec{x}\}.
\end{align*}
Indeed, it is unnecessary to retain any other terms in the summation \eqref{smaller-rhs}, since $M_i$ has a vanishing action on the observable for all other values of $i$. Let us simplify \eqref{smaller-rhs} further. Clearly, when sites $i$ and $i+1$ of $\mu$ are occupied by particles, regardless of their types, $M_i$ has no effect on the set $\vec{z}$, and hence acts directly on $t^{-\chi(\vec{x},\vec{y})}$. Therefore, using \eqref{global2-1}, we find that
\begin{align}
\label{eq-lhs-2}
\sum_{i\in d_1(\vec{x},\vec{y})\bigcup d_2(\vec{x},\vec{y})}
M_i[H(\nu,\cdot)](\mu)
=
H(\nu,\mu^{*})
\sum_{i\in d_1(\vec{x},\vec{y})\bigcup d_2(\vec{x},\vec{y})}
M_i[t^{-\chi(\cdot)}](\vec{x},\vec{y}).
\end{align}
One can now easily show that the right hand side of \eqref{eq-lhs-2} vanishes. To see this, note that when $i\in d_1(\vec{x},\vec{y})$ (namely, when $\mu_i=1$ and $\mu_{i+1}=2$), we have
\begin{align}
\label{eq-lhs-4}
M_i[t^{-\chi(\cdot)}](\vec{x},\vec{y})
=
t\cdot t^{-\chi(\vec{x},\vec{y})-1}
-t^{-\chi(\vec{x},\vec{y})}
=
0.
\end{align}
Similarly, when $i\in d_2(\vec{x},\vec{y})$ (namely, when $\mu_i=2$ and $\mu_{i+1}=1$), we have
\begin{align}
\label{eq-lhs-3}
M_i[t^{-\chi(\cdot)}](\vec{x},\vec{y})
=
t^{-\chi(\vec{x},\vec{y})+1}
-
t\cdot t^{-\chi(\vec{x},\vec{y})}
=
0.
\end{align}
Combining \eqref{smaller-rhs}--\eqref{eq-lhs-4}, we conclude that
\begin{align}
\label{RHS}
\sum_{i \in \mathbb{Z}}
M_i [H(\nu,\cdot)]\left(\mu\right)
=
\sum_{i\in\{z,z+l\}}M_i[H(\nu,\cdot)](\mu),
\end{align}
reducing the action of the generator to just the two sites $z$ and $z+l$. By assumption, 
$\mu_z = \mu_{z+l+1} = 0$, meaning that $\chi(\vec{x},\vec{y})$ is invariant under the action of both $M_z$ and $M_{z+l}$ (since the number of $1$ and $2$-particle crossings will be preserved). This allows us to rewrite \eqref{RHS} as 
\begin{align}
\label{RHS2}
\sum_{i \in \mathbb{Z}}
M_i [H(\nu,\cdot)]\left(\mu\right)
=
t^{-\chi(\vec{x},\vec{y})}
\sum_{i \in \{z,z+l\}}
M_i\left[H(\nu,\cdot)\right](\mu^{*}),
\end{align}
in which the final expression is a purely rank-one quantity.

\medskip
Turning to the left hand side of \eqref{aim-globaleqn}, we use \eqref{global2-1} to write
\begin{align}
\label{eq-rhs-1}
\sum_{i \in \mathbb{Z}} 
L_i [H\left(\cdot,\mu\right) ](\nu)
=
t^{-\chi(\vec{x},\vec{y})}
\sum_{i \in \mathbb{Z}}
L_i \left[H(\cdot,\mu^{*})\right] (\nu)
=
t^{-\chi(\vec{x},\vec{y})}
\sum_{i \in \mathbb{Z}}
M_i\left[H(\nu,\cdot)\right](\mu^{*}),
\end{align}
where the second equality is deduced from the rank-one duality relation of Proposition \ref{prop:7.1}. The final term in \eqref{eq-rhs-1} can be simplified further, since $M_i$ has a vanishing action on the rank-one observable for any $i \not= z,z+l$. Hence,
\begin{align}
\label{eq-rhs-2}
\sum_{i \in \mathbb{Z}} 
L_i [H\left(\cdot,\mu \right) ](\nu)
=
t^{-\chi(\vec{x},\vec{y})}
\sum_{i \in \{z,z+l\}}
M_i\left[H(\nu,\cdot)\right](\mu^{*}).
\end{align}
Comparing \eqref{RHS2} and \eqref{eq-rhs-2} yields the proof of \eqref{aim-globaleqn} in the case of one particle cluster. A generic configuration $\vec{z}(\mu) = \vec{x}(\mu) \cup \vec{y}(\mu)$ can be written as a union of clusters, which then provides a natural splitting of the generator 
$\sum_{i \in \mathbb{Z}} M_i$ into finite disjoint pieces of the form \eqref{smaller-rhs}. One can apply the preceding logic {\it mutatis mutandis}\/ to each such piece, leading to the proof of \eqref{aim-globaleqn} in full generality.

\end{proof}

\subsection{Generalization to arbitrary rank}
\label{ssec:7.3}

Having arrived at the observable \eqref{glob-rnk2}, valued on the configuration spaces of a rank-one and rank-two process, one can immediately see how to generalize the two-species process to arbitrary rank:

\begin{cor}
Let $\nu$ be an infinite rank-one composition, with $\nu_i \in \{0,1\}$ for all $i \in \mathbb{Z}$. Let $\mu$ be an infinite rank-$r$ composition, with 
\begin{align*}
\vec{x}^{(j)}(\mu)
= 
\Big(x^{(j)}_1 < \cdots < x^{(j)}_{m_j}\Big)
\end{align*}
labelling the positions of its $j$-particles, for all $1 \leq j \leq r$. Define an all-rank extension of the crossing statistic \eqref{chi} as follows:
\begin{align*}
\chi\left(\vec{x}^{(1)},\dots,\vec{x}^{(r)}\right)
:=
\#\left\{x \in \vec{x}^{(i)},\ y \in \vec{x}^{(j)}\ \Big|\ i<j,\ x > y \right\}.
\end{align*}
Then the observable
\begin{align}
\label{glob-rnkr}
H\left(\nu,\mu\right)
=
\prod_{j=1}^{r}
\left(
\prod_{x \in \vec{x}^{(j)}(\mu)}
\prod_{i \leq x}
\left(
t^{\nu_i}
\right)
\right)
\cdot
t^{-\chi\left(\vec{x}^{(1)},\dots,\vec{x}^{(r)}\right)}
\end{align}
satisfies the equation
\begin{align*}
\sum_{i \in \mathbb{Z}} 
L_i[H\left(\cdot,\mu\right)](\nu) 
= 
\sum_{i \in \mathbb{Z}}
M_i[H(\nu,\cdot)]\left(\mu\right),
\end{align*}
where $L_i$ and $M_i$ are given by \eqref{eq:Lpsi} and \eqref{eq:Mpsi}, respectively.
\end{cor}

\begin{proof}
One defines the union of all particle positions, 
\begin{align*}
\vec{z}(\mu) = \vec{x}^{(1)}(\mu) \cup \cdots \cup \vec{x}^{(r)}(\mu),
\end{align*}
and proceeds along similar lines as in the proof of Proposition \ref{prop:7.2}, considering firstly the case in which $\vec{z}$ is a single cluster. None of the steps are substantively changed; the sole exception being that the sets $d_1$ and $d_2$ used in \eqref{smaller-rhs} should be replaced by the sets
\begin{align*}
d_{<}\left( \vec{x}^{(1)},\dots,\vec{x}^{(r)} \right)
&=
\left\{x \in \vec{x}^{(i)} \Big| x+1 \in \vec{x}^{(j)}, \ i<j \right\},
\\
d_{>}\left( \vec{x}^{(1)},\dots,\vec{x}^{(r)} \right)
&=
\left\{x \in \vec{x}^{(i)} \Big| x+1 \in \vec{x}^{(j)}, \ i>j \right\},
\end{align*}
respectively.
\end{proof}

\newcommand\arxiv[1]{
\href{http://arxiv.org/abs/#1}{\tt arXiv:#1}}

\end{document}